\documentclass[12pt]{l4dc2022} 

\title[Sampling-based reachability analysis]{A Simple and Efficient Sampling-based Algorithm\\ for General Reachability Analysis}%

\usepackage{times}

\usepackage{url} 
\usepackage{wrapfig}

\usepackage{hyperref}
\usepackage{graphicx}
\usepackage{amsmath}
\usepackage{bm}
\usepackage{mathrsfs}

\def\Inf{\operatornamewithlimits{inf\vphantom{p}}}

\usepackage[font={small,it}]{caption}

\usepackage{color}

\usepackage{cleveref} 

\usepackage{xspace}

\usepackage{enumitem} 

\usepackage{cite} 

\usepackage{amssymb}
\usepackage{algorithm}
\makeatletter
\renewcommand*{\ALG@name}{Alg.}
\makeatother
\usepackage[noend]{algpseudocode}

\newtheorem{thm}{Theorem}

\newtheorem{cor}{Corollary}

\newcommand\rev[1]{#1}

\newcommand\mydots{\hbox to 1em{.\hss.\hss.}}

\newcommand{\comp}{\mathsf{c}} 
\newcommand{\dd}{\textrm{d}} 
\newcommand{\hull}{\textrm{H}} 
\newcommand{\Int}{\textrm{Int}} 
 
\newcommand{\Prob}{\mathbb{P}}

\newcommand{\K}{\mathcal{K}} 

\newcommand{\A}{\mathcal{A}}
\newcommand{\G}{\mathcal{G}}

\newcommand{\B}{\mathcal{B}}

\newcommand{\X}{\mathcal{X}}
\newcommand{\Y}{\mathcal{Y}}

\newcommand{\Xsafe}{\mathcal{X}_{\text{free}}}
\newcommand{\Xgoal}{\mathcal{X}_{\text{goal}}}
\newcommand{\U}{\mathcal{U}}

\newcommand{\R}{\mathbb{R}}

\newcommand{\pinn}{\pi_{\textrm{nn}}}

\newcommand{\randup}{\textsc{RandUP}\xspace}
\newcommand{\gotube}{\textsc{GoTube}\xspace}
\newcommand{\reachsdp}{\textsc{ReachSDP}\xspace}
\newcommand{\reachlp}{\textsc{ReachLP}\xspace}

\newtheorem{myassumption}{Assumption}

\usepackage{multirow}

\newcommand\SmallMatrix[1]{{%
  \tiny\arraycolsep=0.5\arraycolsep\ensuremath{\begin{bmatrix}#1\end{bmatrix}}}}

\author{%
 \Name{Thomas Lew}$^1$ \Email{thomas.lew@stanford.edu}
 \AND
 \Name{Lucas Janson}$^2$ \Email{ljanson@fas.harvard.edu}
 \AND
 \Name{Riccardo Bonalli}$^3$ \Email{riccardo.bonalli@l2s.centralesupelec.fr}
 \AND
 \Name{Marco Pavone}$^1$ \Email{pavone@stanford.edu}\\[1mm]
 \addr $^1$Department of Aeronautics and Astronautics, Stanford University\\
 \addr $^2$Department of Statistics, Harvard University\\
 \addr $^3$Laboratory of Signals and Systems, University of Paris-Saclay, CNRS, CentraleSupélec%
\vspace{-2mm}
}

\begin{document}

\maketitle

\vspace{-1mm}

\begin{abstract}
In this work, we analyze an efficient sampling-based algorithm for general-purpose reachability analysis, which remains a notoriously challenging problem with applications ranging from neural network verification to safety analysis of dynamical systems. 
By sampling inputs, evaluating their images in the true reachable set, and taking their $\epsilon$-padded convex hull as a set estimator, this algorithm applies to general problem settings and is simple to implement. 
Our main contribution is the derivation of asymptotic and finite-sample accuracy guarantees using random set theory. 
This analysis informs algorithmic design to obtain an $\epsilon$-close reachable set approximation with high probability, 
provides insights into which reachability problems are most challenging, 
and motivates safety-critical applications of the technique. %
On a neural network verification task, we show that %
this approach is more accurate and significantly faster than prior work. 
Informed by our analysis, we also design a robust model predictive controller that we demonstrate in hardware experiments.
\end{abstract}

\begin{keywords}%
reachability analysis, random set theory, robust control, neural network verification.
\end{keywords}


\section{Introduction}\label{sec:intro}

\begin{wrapfigure}{!R}{0.44\linewidth}
	\begin{minipage}{0.95\linewidth}
	\vspace{-12.5mm}
\includegraphics[width=1\linewidth,trim=0 0 0 120, clip]{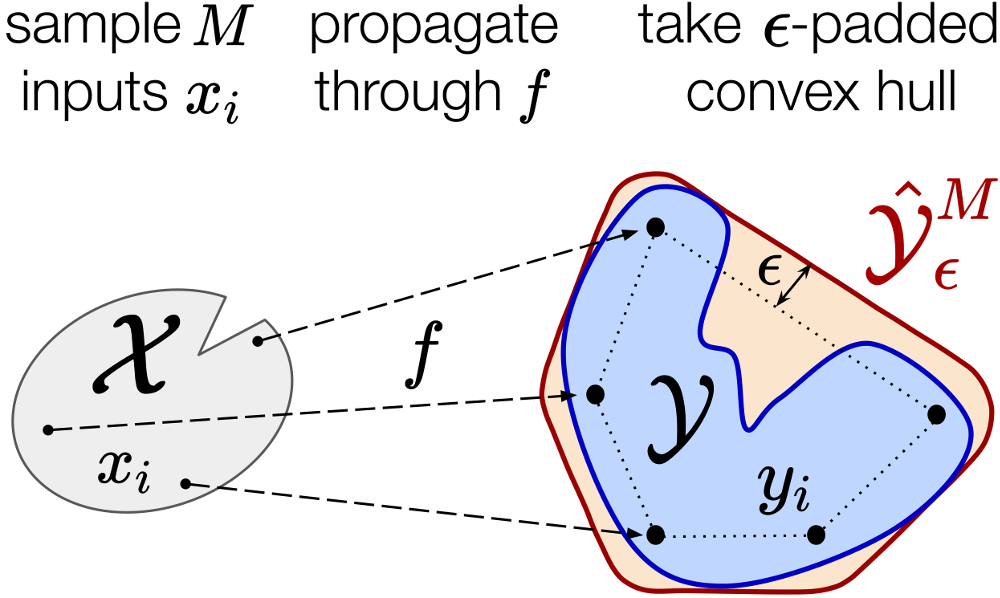}
    \vspace{-5mm}
	\caption{$\epsilon$-\randup consists of three simple steps: 1) sampling $M$ inputs $x_i$ in $\X$, 2) propagating these inputs through the reachability map $f$, and 3) taking the $\epsilon$-padded convex hull $\hat\Y_\epsilon^M$ to approximate the reachable set $\Y$. %
	}\label{fig:randup}
	\vspace{-5mm}
	\end{minipage}%
	\end{wrapfigure}

Forward reachability analysis entails characterizing the reachable set of outputs of a given function corresponding to a set of inputs.
This type of analysis underpins a plethora of applications in model predictive control, neural network verification, and safety analysis of dynamical systems. 
Sampling-based reachability analysis techniques are a particularly simple class of methods to implement; however, conventional wisdom suggests that if insufficient representative samples %
are considered, these methods may not be robust in that they cannot rule out edge cases missed by the sampling procedure. 
Alternatively, by leveraging structure in specific problem formulations or computational methods designed for exhaustivity (e.g., branch and bound), a large range of algorithms with deterministic accuracy and performance guarantees have been developed. However, these methods often sacrifice simplicity and generality for their power, motivating the development of algorithms that avoid such restrictions.

In this work, we analyze a simple yet efficient sampling-based algorithm for general-purpose reachability analysis. 
As depicted in Figure \ref{fig:randup}, it consists of 1) sampling inputs, 2) propagating these inputs, and 3) taking the padded convex hull of these output samples.   
We refer to this \textsc{Rand}omized \textsc{U}ncertainty \textsc{P}ropagation algorithm as $\epsilon$-\randup: it is simple to implement, 
benefits from statistical accuracy guarantees, 
and
applies to a wide range of problems including reachability analysis of uncertain dynamical systems with neural network controllers. 
Importantly, $\epsilon$-\randup fulfills key desiderata that a general-purpose reachability analysis algorithm should satisfy:
\begin{itemize}[leftmargin=5mm]
\setlength\itemsep{0mm}
\vspace{-1mm}
\item it works with any choice of possibly nonlinear reachability maps and non-convex input sets,
\item its estimate of the reachable set is conservative with high probability and tighter than prior work,
\item it is efficient and does not require precomputations, which is a key advantage for learning-based control applications where uncertainty bounds and models are updated in real-time. 
\end{itemize}
Our main contribution is a thorough analysis of the statistical properties of $\epsilon$-\randup. 
Specifically:
\begin{enumerate}[leftmargin=5mm]
\setlength\itemsep{0mm}
\vspace{-1mm}
\item %
We prove that the %
set estimator converges to the $\epsilon$-padded convex hull of the true reachable set as the number of samples increases. %
Our assumption about the sampling distribution is weaker than in related work and implies that sampling the boundary of the input set is sufficient. %
This asymptotic result justifies using %
$\epsilon$-\randup %
as a thrustworthy baseline for offline validation whenever the reachability map and the input set are complex and no tractable algorithm %
exists. 

\item %
We derive a finite-sample bound for the Hausdorff distance between the output of $\epsilon$-\randup and the convex hull of the true reachable set, %
assuming that the reachability map is Lipschitz continuous. 
This result informs algorithmic design (e.g., how to choose the number of samples to obtain an $\epsilon$-accurate approximation with high probability), 
sheds insights into %
which %
problems are most challenging, 
and motivates using this simple algorithm in safety-critical applications. 
\end{enumerate}
\vspace{-1mm}
We demonstrate $\epsilon$-\randup on a neural network controller verification task and show that it is highly competitive with prior work. We also embed this algorithm within a robust model predictive controller and present hardware results demonstrating the reliability of the approach.

\vspace{-1mm}

\section{Related work}\label{sec:related_work}

Reachability analysis has found a wide range of applications ranging from model predictive control \citep{Schurmann2018}, 
robotics \citep{Shao2021,LewEtAl2022}, 
neural network verification \citep{Tran2019,Hu2020},  
to orbital mechanics \citep{Wittig2015}. 
Reachability analysis is particularly relevant in safety-critical applications which require the strict satisfaction of specifications. 
For instance, 
a drone transporting a package should never collide with obstacles and respect velocity bounds for any payload mass in a bounded input set. %
In contrast to stochastic problem formulations which typically consider the inputs as random variables with known probability distributions \citep{Webb2019,Sinha2020,DevonportL4DC2020}, 
we consider robust formulations which are of %
interest whenever minimal information about the inputs is available. 

Deterministic %
algorithms are often tailored to the particular parameterization of the reachability map and to the shape of the input set. For instance, one finds methods that are particularly designed for neural networks \citep{Tran2019,IvanovVerisig2019,Hu2020}, 
nonlinear hybrid systems \citep{Chen2013,Kong2015}, 
linear dynamical systems with zonotopic \citep{Girard2005} and ellipsoidal \citep{Kurzhanski2000} parameter sets, 
etc.  
We refer to \citep{Liu2021} and \citep{Althoff2021} for recent comprehensive surveys. 
Such algorithms %
have deterministic accuracy guarantees but require problem-specific structure that restricts the class of systems 
they apply to. 
Given the wide range of applications of reachability analysis, there is a pressing need for the development and analysis of simple algorithms that can be applied to general problem formulations. %

On the other hand, sampling-based algorithms reconstruct the reachable set %
from sampled outputs. %
The stochasticity is typically controlled by the engineer, who selects the number of samples and their distribution. 
A key strength of this methodology is the possible use of black-box models %
with arbitrary input sets, %
which allows using complex simulators of the system. 
For instance, kernel-based methods \citep{DeVito2014,Rudi2017,ThorpeL4DC2021} have been proposed as a strong approach for data-driven reachability analysis. Kernel-based methods are highly expressive, as selecting a completely separating kernel \citep{DeVito2014} enables reconstructing any closed set %
to arbitrary precision given enough samples. 
Their main drawback %
is the potentially expensive evaluation of %
the estimator for a large number of samples. %
Its implicit representation %
as a level set is also not particularly convenient for downstream applications. %

Sampling-based reachable set estimators with pre-specified shapes have been proposed to simplify computations and downstream applications. 
Recently, \citep{LewPavone2020} proposed to approximate %
reachable sets %
with the convex hull of the samples,  
but this approach is not guaranteed to return a conservative approximation. %
Ellipsoidal and rectangular sets are computed in \citep{DevonportL4DC2020} using the scenario approach, but this work tackles a different problem formulation with inputs that are random variables with known distribution.  
To tackle the robust reachability analyis problem setting, %
\citep{Gruenbacher2021} use a ball estimator %
that bounds the samples. %
The statistical analysis is restricted to ball-parameterized input sets, %
uniform sampling distributions, and smooth diffeomorphic reachability maps %
that represent the solution of a neural ordinary differential equation \citep{Chen2018} from the input set. %
In practice, using an outer-bounding ball %
is more conservative than taking the convex hull of the samples, see %
Section \ref{sec:results}.

In this work, we slightly modify \randup \citep{LewPavone2020} with an additional $\epsilon$-padding step to yield finite-sample outer-approximation guarantees,  %
Our analysis leverages  
random set theory \citep{Matheron1975,Molchanov_BookTheoryOfRandomSets2017}, which provides a natural mathematical framework to analyze the reachable set estimator. %
We characterize its accuracy using the Hausdorff distance to the convex hull of the true reachable set, which provides an intuitive error measure %
that can be directly used for downstream control applications. 
Our analysis draws inspiration from the vast literature on statistical geometric inference, which proposes different %
set estimators including  
union of balls \citep{Devroye1980,Baillo2001}, 
convex hulls \citep{RipleyPoissonForest1977,Schneider1988,Dumbgen1996}, 
$r$-convex hulls \citep{Rodriguez2016,RodriguezCasal2019,AriasCastro2019},
Delaunay complexes \citep{Boissonnat2013,AamariPhD2017,Aamari2018}, 
and kernel-based estimators \citep{DeVito2014,Rudi2017}. 
This research typically makes assumptions about the set to be reconstructed %
(e.g., %
it is convex \citep{Dumbgen1996} or has bounded reach \citep{Cuevas2009}) and considers points that are directly sampled from this set. 
In this work, we derive similar results for reachable sets 
given known properties of the 
input set,  
reachability map, and 
chosen input sampling distribution.

\vspace{-1mm}
\section{Problem definition}\label{sec:formal_setting}
\vspace{-1mm}
In this section, we introduce our notations and  problem formulation. %
Due to space constraints, we leave measure-theoretic details to Appendix \ref{apdx:random_set_theory}.  
We denote 
$\lambda(\cdot)$ for the Lebesgue measure over $\R^p$, 
$\Gamma(\cdot)$ for the %
gamma function, 
$\hull(A)$ for the convex hull of a subset $A\subset\R^n$, $A^\comp=\R^n\setminus A$ for its complement, 
$\partial A$ for its boundary, 
$\oplus$ for the Minkowski sum, 
$B(x,r)\,{:=}\,\{y\,{\in}\,\R^n{:}\, 
 \|y\,{-}\,x\|\,{\leq}\, r
\}$ for the closed ball of center $x\in\R^n$ and radius $r\geq 0$, 
and $\mathring{B}(x,r)$ for the open ball. 
The family of nonempty compact subsets of $\R^n$ is denoted as $\K$.  
For any $A\in\K$ and $d>0$, $D(A, d)\,{:=}\,\min\{n\,{\in}\,\mathbb{N} : 
\exists \{a_1,{\mydots},a_n\}\,{\subset}\,\R^n, \ 
A\,{\subset}\, B(a_1,d)\,{\cup}\,{\mydots}\,{\cup}\, B(a_n,d)
\}$ denotes the $d$-covering number of $A$. 

Let $\X\subset\R^p$ be a compact nonempty set of inputs and  
$f:\R^p\rightarrow\R^n$ be a continuous function. 
In this work, we tackle the general problem of reachability analysis, i.e., %
characterizing the set of reachable outputs $y=f(x)$ for all possible inputs $x\in\X$. 
This problem is also often referred to as uncertainty propagation. %
Mathematically, the objective consists of efficiently computing an accurate approximation of the reachable set $\Y\subset\R^n$, which is defined as 
\begin{align}
\label{eq:reach_set}
\Y = f(\X) =
\{ 
	f(x) \, :\, x\in\X
\}
.
\end{align}
To tackle this problem, 
$\epsilon$-\randup relies on the choice of three parameters: 
a number of samples $M\in\mathbb{N}$, 
a padding constant $\epsilon>0$, 
and 
a sampling distribution $\Prob_\X$ on measurable subsets of $\R^{p}$. 
As depicted in Figure \ref{fig:randup}, %
$\epsilon$-\randup consists of sampling $M$ independent identically-distributed inputs $x_i$ in $\X$ according to $\Prob_\X$, 
of evaluating each output $y_i=f(x_i)$, %
and 
of computing the $\epsilon$-padded convex hull 
\begin{equation}\label{eq:estimator_eps}
\hat\Y_\epsilon^M:=\hull\left(\{y_i\}_{i{=}1}^M\right)\oplus B(0,\epsilon).
\end{equation}
Our analysis hinges on the observation that the reachable set estimator $\hat\Y_\epsilon^M$ is a \textit{random compact set}, i.e., $\hat\Y_\epsilon^M$ is a random variable taking values in the family of nonempty compact sets $\K$.  
We refer to Appendix \ref{apdx:random_set_theory} for rigorous definitions using random set theory. 
Intuitively, different input samples $x_i$ in $\X$ induce different output samples $y_i$ in $\Y$, resulting in  different approximated reachable sets $\hat\Y_\epsilon^M$.  
To characterize the accuracy of
the estimator,
we use the \textit{Hausdorff metric}, which is defined as 
\begin{equation}\label{eq:metric:Hausdorf}
d_\textrm{H}(A,B)
:=
\max\big(
\sup_{x\in B}
\Inf_{y\in A}
\|x-y\|, \
\sup_{x\in A}
\Inf_{y\in B}
\|x-y\|
\big)
\quad
\text{for any $A,B\in\K$.}
\end{equation}
This metric induces a topology and an associated $\sigma$-algebra, which enables 
rigorously defining random compact sets as random variables and  describing their convergence; see Appendix \ref{apdx:random_set_theory}. 
Interestingly, the distribution of a random compact set is characterized by the probability that it intersects any given compact set. 
We use this fact in Sections \ref{sec:asymptotic} and \ref{sec:finite_sample}, where we %
characterize the probability that the set estimator $\hat\Y_\epsilon^M$ intersects well-chosen sets along the boundary of the true reachable set. 
By analyzing the distribution of $\hat\Y_\epsilon^M$, 
this approach allows bounding the Hausdorff distance between $\hat\Y_\epsilon^M$ and the convex hull of the true reachable set $\hull(\Y)$ with high probability.

\vspace{-1mm}

\section{Asymptotic analysis}\label{sec:asymptotic}

In this section, we provide an asymptotic analysis under minimal assumptions about the input set and the reachability map (namely, that $\X$ is compact and $f$ is continuous). %
To enable the reconstruction of the true convex hull $\hull(\Y)$ using the sampling-based set estimator $\hat\Y_\epsilon^M$, we make one assumption about the sampling distribution $\Prob_\X$ for the inputs $x_i$. Note that by definition, $\Prob_\X(\X)=1$.

\begin{myassumption}\label{assum:XBoundary:posMeasure} 
$\Prob_\X(\{x\in\X: f(x)\in\mathring{B}(y,r)\})>0$  
for all $y\in\partial\Y$ and all $r>0$.
\end{myassumption}
This assumption states that the probability of sampling an output arbitrarily close to any point on the boundary of the true reachable set is strictly positive. 
In other words, the boundary of the reachable set should be contained in the support of the distribution of the output samples $y_i$.  
Assumption \ref{assum:XBoundary:posMeasure} is weaker than the associated assumption in \citep[Theorem 2]{LewPavone2020}, which can be restated as ``\textit{$\Prob_\X(f^{-1}(A))>0$ for any open set $A\subset\R^n$ such that $\Y\cap A\neq \emptyset$}''. 
Indeed, Assumption \ref{assum:XBoundary:posMeasure} only considers open neighborhoods of the boundary $\partial\Y$, as opposed to all open sets intersecting $\Y$. 
Selecting a sampling distribution $\Prob_\X$ that satisfies Assumption \ref{assum:XBoundary:posMeasure} is easy. For instance, if $\X$ has a smooth boundary (see Assumption \ref{assum:Theta:r_convex}), then the uniform distribution over $\X$ satisfies Assumption \ref{assum:XBoundary:posMeasure}.

Assumption \ref{assum:XBoundary:posMeasure} is sufficient to prove that the random set estimator $\hat\Y_\epsilon^M$ converges to the $\epsilon$-padded convex hull of $\Y$ as the number of samples $M$ increases. 
Below, we prove a more general result which allows for variations of the padding radius $\epsilon$ as the number of samples increases. 

\begin{thm}[Asymptotic Convergence]\label{thm:asymptotic_convergence}
Let $\bar{\epsilon}\geq 0$ and 
$(\epsilon_M)_{M\in\mathbb{N}}$ be a sequence of padding radii such that $\epsilon_M\geq 0$ for all $M\in\mathbb{N}$ and $\epsilon_M\rightarrow \bar{\epsilon}$ as $M\rightarrow\infty$. 
For any $\epsilon\geq 0$, define the estimator $\hat\Y^M_{\epsilon}=\hull\left(\{y_i\}_{i{=}1}^M\right)\oplus B(0,\epsilon)$.  
Then, under Assumption \ref{assum:XBoundary:posMeasure}, 
almost surely, 
as $M\rightarrow\infty$, 
$$
d_H(
\hat\Y_{\epsilon_M}^M, 
\hull(\Y)\oplus B(0,\bar{\epsilon})
)
\mathop{\longrightarrow} 0.
$$
\end{thm}

\begin{proof}
We refer to Appendix \ref{apdx:proof:thm:asymptotic_convergence}. 
We leverage \citep[Proposition 1.7.23]{Molchanov_BookTheoryOfRandomSets2017} which states sufficient conditions for the convergence of random compact sets %
and use properties of the convex hull to relax the corresponding assumption in \citep{LewPavone2020} with Assumption  \ref{assum:XBoundary:posMeasure}. 
\end{proof}
\vspace{-3mm}
Practically, Theorem \ref{thm:asymptotic_convergence} justifies using $\epsilon$-\randup for general continuous maps $f$ and compact sets $\X$. 
This consistency result implies that choosing any converging sequence of padding radii (e.g., $\epsilon_M=1/M$) guarantees the convergence of the random set estimator $\hat\Y_{\epsilon_M}^M$  to the $\bar{\epsilon}$-padded convex hull of the true reachable set. 
As a particular case, selecting a constant padding radius $\epsilon$ (which yields $\epsilon$-\randup) guarantees that $\Y_{\epsilon}^M$ converges to the $\epsilon$-padded convex hull $\hull(\Y)\oplus B(0,\epsilon)$.

Compared to \citep[Theorem 2]{LewPavone2020}, which only treats the case with constant zero padding radii $\epsilon_M=\bar\epsilon=0$ (i.e., without $\epsilon$-padding the convex hull of the output samples), Theorem \ref{thm:asymptotic_convergence} allows for variations of the padding radii $\epsilon_M$ and is proved under weaker assumptions. 
Instead of relying on %
$\epsilon$-covering arguments (e.g., see Corollary 1 in \citep{Dumbgen1996} which %
assumes that $\Y$ is convex), we use \citep[Proposition 1.7.23]{Molchanov_BookTheoryOfRandomSets2017}  to conclude asymptotic convergence. 
This proof technique allows deriving a general result that does not depend on the exact sampling density along the boundary $\partial\Y$ and uses a sequence of padding radii $\epsilon_M$ converging arbitrarily slowly to some constant $\bar{\epsilon}\geq 0$. 

\section{Finite-sample analysis}\label{sec:finite_sample}

Theorem \ref{thm:asymptotic_convergence} provides asymptotic convergence guarantees that support the application of $\epsilon$-\randup in general scenarios (e.g., as a baseline for offline validation in complex problem settings), but does not provide finite-sample guarantees which are of practical interest in safety-critical applications. 
Deriving stronger statistical guarantees requires leveraging more information about the structure of the problem. We derive finite-sample rates under general assumptions in Section \ref{sec:finite_sample:general} and analyze a particular case in Section \ref{sec:finite_sample:rconvex}. 
We discuss practical implications of our results in Section \ref{sec:finite_sample:insights}.

\subsection{General finite-sample statistical guarantees}\label{sec:finite_sample:general} 

To derive convergence rates and outer-approximation guarantees given a finite number of samples $M$, 
we first make an assumption   
about the smoothness of the reachability map $f$.

\begin{myassumption}\label{assum:f:lipschitz}
The reachability map $f:\R^p\rightarrow\R^n$ is $L$-Lipschitz: for some constant $L\geq 0$, 
$\|f(x_1)-f(x_2)\|
\leq
L\,\|x_1-x_2\|\  \
\text{for all }\ x_1,x_2\in\X$.
\end{myassumption}
Next, we make an assumption about the sampling distribution $\Prob_\X$ along the input set boundary $\partial\X$.

\begin{myassumption}\label{assum:sampling_density}
Given $\epsilon,L\,{>}\,0$, there exists $\Lambda_{\epsilon}^{L}\,{>}\,0$ such that $\Prob_\X\left(B\left(x,\frac{\epsilon}{2L}\right)\right)\,{\geq}\, \Lambda_{\epsilon}^{L}$ for all  
	$x\in\partial\X$.
\end{myassumption}
Given any boundary input $x\in\partial\X$, 
the constant $\Lambda_{\epsilon}^{L}$ characterizes the probability of sampling an input $x_i$ that is $\epsilon/(2L)$-close to $x$. 
Selecting a sampling distribution that satisfies Assumption \ref{assum:sampling_density} is simple; we provide examples in Sections \ref{sec:finite_sample:rconvex} and \ref{sec:results}. 
As we show next, these two assumptions are sufficient to derive finite-sample convergence rates for $\epsilon$-\randup. 
Recall that $D(\partial\X, d)$ denotes the $d$-packing number of $\partial\X$, which is necessarily finite by the compactness of $\X$.

\begin{thm}[Finite-Sample Bound] \label{thm:conservative_finite_sample} 
Define the estimator $\hat\Y^M=\hull\left(\{y_i\}_{i{=}1}^M\right)$ and the probability threshold
$
\delta_M= 
D(\partial\X,\epsilon/(2L))(1 -
	\Lambda_{\epsilon}^{L}
)^M
$.
Then, under Assumptions \ref{assum:f:lipschitz} and \ref{assum:sampling_density} 
\rev{and assuming that $\partial\Y\subseteq f(\partial\X)$}, with probability at least $1-\delta_M$, 
$$
d_H(
\hat\Y^M, 
\hull(\Y)
)\leq \epsilon
\quad\text{and}\quad
\Y\subseteq\hat\Y_\epsilon^M. 
$$
\end{thm}
\begin{proof}
We refer to Appendix \ref{apdx:proof:thm:conservative_finite_sample} for a complete proof.
\end{proof}
Using a similar analysis, 
one could derive convergence rates for the $\epsilon$-padded union of balls estimator %
\citep{Devroye1980,Baillo2001} that would depend on the $\epsilon$-covering number of the entire input set $D(\X,\epsilon)$. 
In the general case, %
$D(\partial\X,\epsilon)\leq D(\X,\epsilon)$: %
Theorem \ref{thm:conservative_finite_sample} indicates that 
using a convex hull is more sample-efficient than a union of balls \rev{(assuming that $\partial\Y\subseteq f(\partial\X)$, see Appendix \ref{apdx:proof:thm:conservative_finite_sample} for further details)}. It is better suited if $\Y$ is convex or if an approximation of $\hull(\Y)$ is sufficient for the downstream application, as is usual in control applications which typically use convex reachable set approximations, see \citep{LewPavone2020}.  

\subsection{Analysis of a particular setting: smooth input set and continuous distribution}\label{sec:finite_sample:rconvex}

	\begin{wrapfigure}{R}{0.33\linewidth}
	\begin{minipage}{0.95\linewidth}
    \centering
    \vspace{-1mm}
 \includegraphics[width=0.95\linewidth,trim=0 40 0 0, clip]{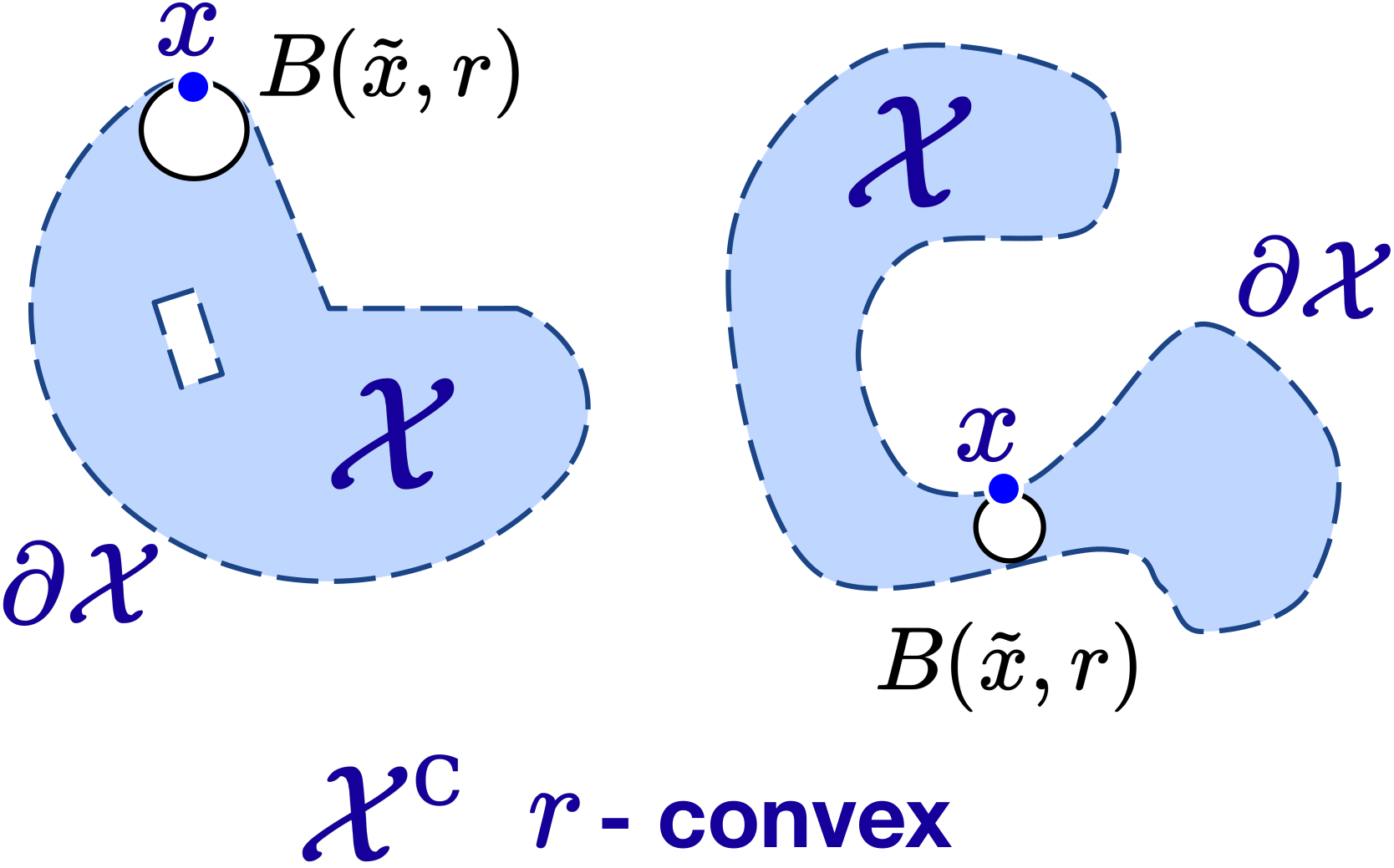}
 \centering
 \includegraphics[width=0.95\linewidth,trim=0 70 0 0, clip]{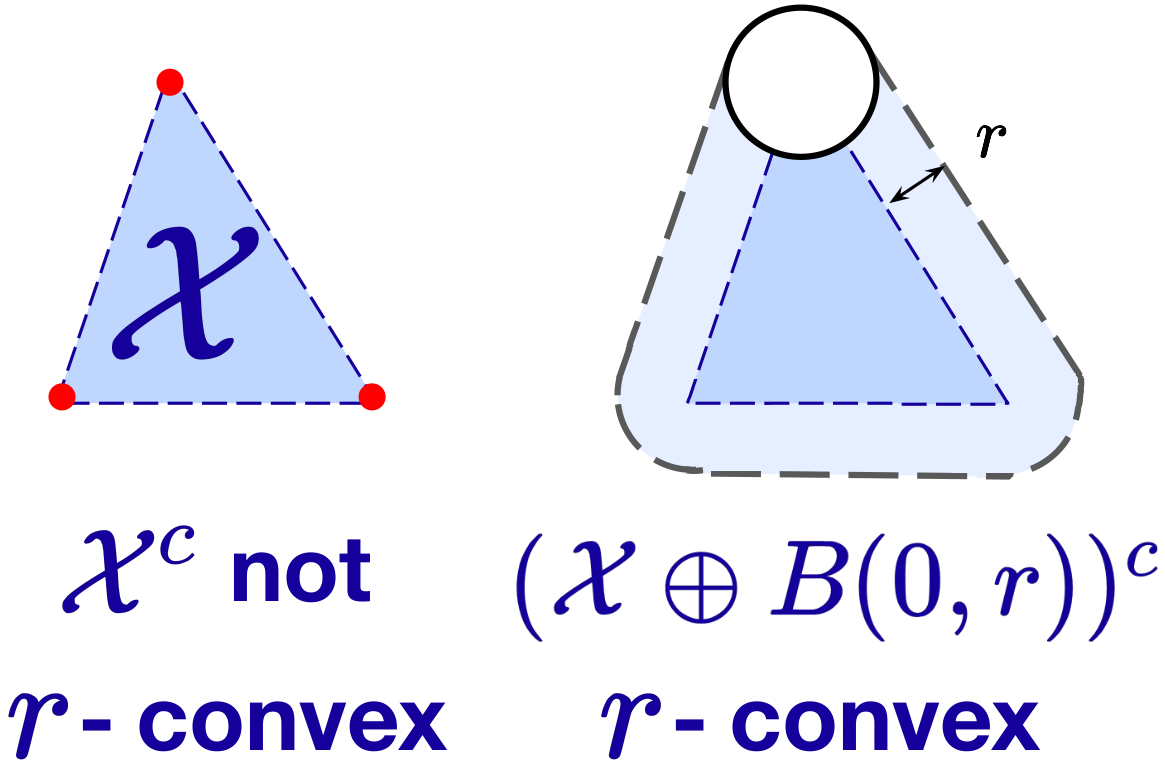}
	\caption{%
  \textbf{Top}: sets $\X$ satisfying Assumption \ref{assum:Theta:r_convex} can be non-convex, have holes, and be disconnected. 
  \textbf{Bottom}: if $\X^\comp$ is not $r$-convex, it is still possible to find a conservative approximation that is $r$-convex.}\label{fig:rconvex}
	\vspace{-10mm}
	\end{minipage}%
	\end{wrapfigure}
	
In many applications, the boundary of the input set is smooth (e.g., $\X$ is a $2$-norm ball). 
In this setting, we can apply Theorem \ref{thm:conservative_finite_sample} to derive finite-sample guarantees for general continuous sampling distributions. We state this smoothness assumption below.

\begin{myassumption}\label{assum:Theta:r_convex} 
$\X^\comp$ is $r$-convex for some $r>0$. Equivalently, 
for any $x\in\partial\X$, 
 there exists $\tilde{x}\in\X$ such that 
 $x\in B(\tilde{x},r)\subseteq \X$.
\end{myassumption}
Assumption \ref{assum:Theta:r_convex} guarantees that for any parameter $x$ on the boundary $\partial\X$, one can find a ball of radius $r$ contained in $\X$ that also contains $x$, see Figure \ref{fig:rconvex}. 
This assumption corresponds to a general inwards-curvature condition of the boundary $\partial\X$. 
It is a common assumption in the literature \citep{Walther1997,Rodriguez2016,RodriguezCasal2019,AriasCastro2019} and is related to the notion of reach \citep{Federer1959,Cuevas2009,AamariPhD2017} that bounds the curvature of the boundary $\partial\X$. 
To guarantee its satisfaction, 
one can replace $\X$ with $\X\oplus B(0,r)$ \citep{Walther1997} before performing reachability analysis, which would yield a more conservative estimate of $\Y$. %
Next, we state an assumption about the sampling distribution $\Prob_\X$.

 \begin{myassumption}\label{assum:sampling_density:cor}
$\Prob_\X(A)\geq p_0\lambda(A)$ for all  
 	measurable sets $A\subset\X$ for some constant $p_0>0$.
 \end{myassumption}
This assumption states that the sampling distribution %
admits a lower-bounded continuous density. 
 Specifically, there exists a density function $p_\X:\R^p\rightarrow\R_+$ such that $ 
\Prob_\X(A)
=
\int_{A} p_{\X}(x)\dd x\geq p_0\int_{A} \dd x=p_0\lambda(A)$ for any measurable subset $A\subset\X$. 
For instance, the uniform distribution over $\X$ satisfies this assumption. 
Similarly to Assumption \ref{assum:sampling_density}, this density assumption %
can be relaxed to neighborhoods %
of $\partial\X$; we leave this extension for future work.  
We obtain the following corollary.

\begin{cor}%
\label{cor:conservative_finite_sample} 
 Define the estimator $\hat\Y^M=\hull\left(\{y_i\}_{i{=}1}^M\right)$, 
 the offset vector 
 $\vec{r}=(r,0,\dots,0)\in\R^p$, 
 the volume $\Lambda_{\epsilon}^{r,L}=\lambda\big(
 B(0,\epsilon/(2L)) 
 \cap B(\vec{r},r)
 \big)$, 
 and the threshold
 $
 \delta_M= 
 D(\partial\X,\epsilon/(2L))\smash{(1 -
 	p_0 \Lambda_{\epsilon}^{r,L}
 )^M}
 $.
Then, under Assumptions \ref{assum:f:lipschitz},  \ref{assum:Theta:r_convex} and \ref{assum:sampling_density:cor} 
\rev{and assuming that $\partial\Y\subseteq f(\partial\X)$,}  
with probability at least $ 1-\delta_M$,
$$
d_H(
\hat\Y^M, 
\hull(\Y)
)\leq \epsilon
\quad \text{ and }\quad\ 
\Y\subseteq\hat\Y_\epsilon^M.
$$
\end{cor}
\begin{proof}
We refer to Appendix \ref{apdx:proof:cor:conservative_finite_sample}.
We first prove that Assumptions \ref{assum:Theta:r_convex} and \ref{assum:sampling_density:cor} imply that Assumption \ref{assum:sampling_density} holds with $\Lambda_{\epsilon}^{L}=p_0\Lambda_{\epsilon}^{r,L}$. 
The finite-sample bound then follows by applying Theorem \ref{thm:conservative_finite_sample}. 
\end{proof}

The constant $\Lambda_{\epsilon}^{r,L}$ corresponds to the $p$-dimensional Lebesgue volume of two hyperspherical caps and can be computed analytically, see \citep{Li2011,Petitjean2013} and Appendix \ref{appendix:spherical_caps}.

\subsection{Insights: the difficulty of reachability analysis and  algorithmic design}\label{sec:finite_sample:insights}
Theorem \ref{thm:conservative_finite_sample} reveals which characteristics of the problem make reachability analysis challenging: 
\begin{itemize}[leftmargin=4.5mm]
  \setlength\itemsep{0.0mm}%
    
    \item \textbf{Assuming the smoothness of $f$ is necessary:} %
    given an input set $\X$ and a sampling distribution $\Prob_{\X}$, %
    one can construct problems for which sampling-based reachability analysis algorithms require arbitrarily many samples to compute an $\epsilon$-accurate approximation of $\Y$, see Section \ref{sec:results:sensitivity}. To derive finite-sample rates, assuming that the reachability map $f$ is $L$-Lipschitz (Assumption \ref{assum:f:lipschitz}) is necessary if only  assumptions on input coverage density (Assumption \ref{assum:sampling_density}) are available.
   
    \item \textbf{The smoother the easier}: 
a smaller Lipschitz constant $L$ 
and a larger radius parameter $r$ 
induce 
tighter bounds in Theorem \ref{thm:conservative_finite_sample}, requiring a smaller number of samples $M$ to obtain a desired accuracy with high probability $1-\delta_M$. 
Indeed, such conditions guarantee a lower bound on the probability of sampling outputs $y_i=f(x_i)\in\Y$ that are close to the boundary $\partial\Y$, which is necessary to accurately reconstruct the true convex hull of the reachable set from samples.

    \item \textbf{Scalability}: by Theorem \ref{thm:conservative_finite_sample}, the number of required samples to reach a desired  $\epsilon$-accuracy with high probability depends on the covering number. This constant characterizes the size of the parameter space 
    in terms of dimensionality 
    (the number of different parameters) and volume (variations of each parameter). 
    Given any $\X\,{\in}\,\K$ and $d\,{=}\,\sup_{x\in\partial\X}\|x\|$, a simple and general bound for the covering number is  
    $
    D(\partial\X,\epsilon)
    \,{\leq}\,
    \left(
    2d\sqrt{n}/\epsilon
    \right)^n
    $ \citep{ShalevShwartz2009}.
    
\end{itemize}

\section{Results and applications}\label{sec:results}
We perform a sensitivity analysis in Section \ref{sec:results:sensitivity} to illustrate the insights from Theorem \ref{thm:conservative_finite_sample}. %
In Section \ref{sec:results:verif_closed_nn}, we compute the reachable sets of a dynamical system with a simple neural network policy and compare with prior work. 
Finally, in Section \ref{sec:results:mpc}, we embed  $\epsilon$-\randup in a model predictive control (MPC) framework to reliably control a robotic platform.  %
 Our code and hardware results are available at {\scriptsize \url{https://github.com/StanfordASL/RandUP}} and {\scriptsize \url{https://youtu.be/sDkblTwPuEg}}. 
All computation times are measured on a computer with a 3.70GHz Intel Core i7-8700K CPU.

\newpage
\phantom{asdf}

\vspace{-14mm}

\subsection{Sensitivity analysis}\label{sec:results:sensitivity}

	\begin{wrapfigure}{R}{0.40\linewidth}
	\begin{minipage}{0.95\linewidth}
	\vspace{-11mm}
    \centering
	\includegraphics[width=1\linewidth]{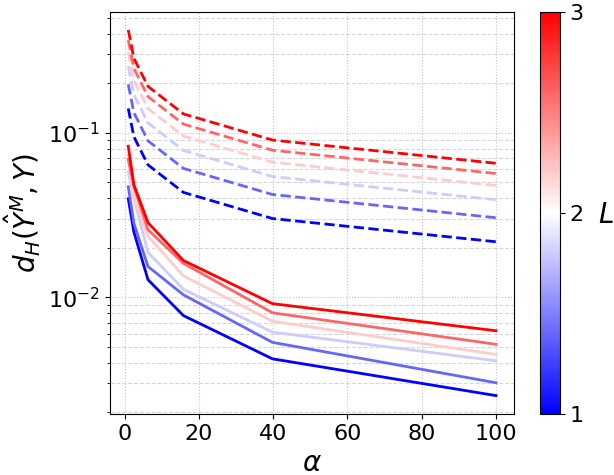}
	\caption{Results for the sensitivity analysis in Section \ref{sec:results:sensitivity}. Experimental results are shown with continuous lines, theorical upper bounds with dashed lines.}\label{fig:sensitivity}
	\vspace{-3mm}
	\end{minipage}%
	\end{wrapfigure}
	
	We analyze the sensitivity of $\epsilon$-\randup  to the sampling distribution and the smoothness of the reachability map. 
	We consider a $2$-dimensional input ball $\X=B(0,1)$ and the map $f(x)=(Lx_1,x_2)$ with $L\geq 1$. Clearly, $\X^\comp$ is $1$-convex and $f$ is $L$-Lipschitz continuous, so %
	Corollary \ref{cor:conservative_finite_sample} applies for any sampling distribution satisfying Assumption \ref{assum:sampling_density:cor}. 
We consider a distribution $\Prob_\X^\alpha$ that depends on a parameter $\alpha\geq 1$, such that $\Prob_\X^\alpha$ varies from a uniform distribution over $\X$ for $\alpha=1$ to a uniform distribution over the boundary $\partial\X$ as $\alpha\rightarrow\infty$. Given $\delta_M=10^{-3}$, we determine the minimum padding $\epsilon$ guaranteeing $\Prob(
d_H(
\hat\Y^M, 
\Y
)\leq \epsilon
)\geq 1-\delta_M$ using Corollary \ref{cor:conservative_finite_sample}, see Appendix \ref{apdx:sensitivity}. 
We take $M=1000$ samples and present results in Figure \ref{fig:sensitivity}. 
We observe better performance than the predicted finite-sample bounds 
and that distributions with a higher probability of sampling close to the boundary %
(i.e., larger values of $\alpha$) perform better, corresponding to lower Hausdorff distance errors. 
Also, $\epsilon$-\randup performs better on problems with smoother reachability maps, %
as is visible from our empirical evaluation and theoretical bounds on the Hausdorff distance.  This validates the discussion in Section \ref{sec:finite_sample:insights}.

\vspace{-2mm}
    
\subsection{Verification of neural network controllers}\label{sec:results:verif_closed_nn}

 \begin{figure}[htb!]
 \vspace{-3mm}
 \centering
  \includegraphics[width=0.95\linewidth]{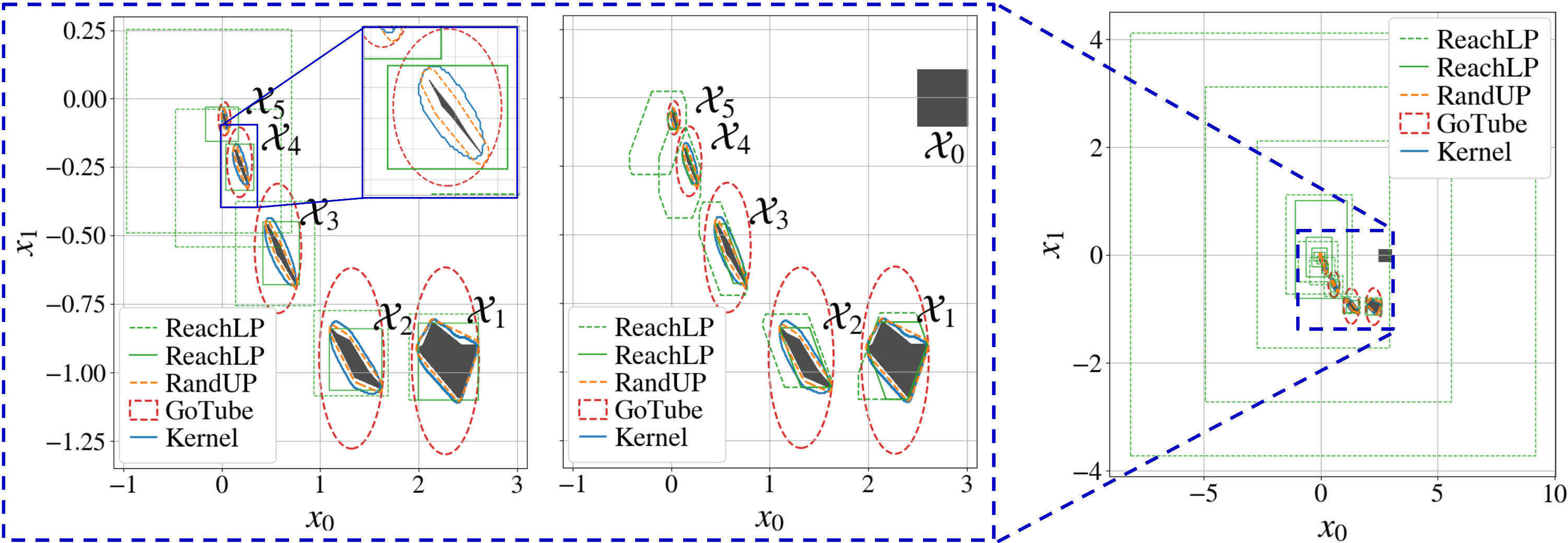}
 \vspace{-3mm}
  \caption{Reachable sets computed in  Section \ref{sec:results:verif_closed_nn} for a total prediction horizon $N=9$. Sets from the formal method \reachlp are shown in green, dashed sets correspond to no input splitting, straight-lines correspond to splitting $\X_0$ into $16$ components.  
  We use $M=10^3$ samples for all sampling-based methods and $\epsilon=0.02$. 
}
 \label{fig:nn_controller:all}
 \vspace{-3mm}
 \end{figure}
Next, we consider the verification of a neural network controller $u_t=\pi_{\textrm{nn}}(x_t)$ for a known linear dynamical system $x_{t+1}=Ax_t+Bu_t$, where $t\in\mathbb{N}$ denotes a time index, and $x_t\in\R^2$ and $u_t\in\R$ denote the state and control input. %
Given a rectangular set of initial states $\X_0\subset\R^2$, 
the problem consists of estimating the reachable set at time $t\in\mathbb{N}$ defined as
$
\X_t=\{(A(\cdot)+B\pinn(\cdot))\circ\dots\circ 
(Ax_0+B\pi_{\textrm{nn}}(x_0)): x_0\in\X_0\}$. 
Defining $(\X,\Y)=(\X_0,\X_t)$ and $f(x)=(A(\cdot)+B\pinn(\cdot))\circ\dots\circ 
(Ax+B\pi_{\textrm{nn}}(x))$, we see that this problem fits the mathematical form described in Section \ref{sec:intro}. 
We use a ReLU network $\pinn$ from \citep{Everett21_journal} with two layers of $5$ neurons each. 
	\begin{wrapfigure}{R}{0.40\linewidth}
	\begin{minipage}{0.95\linewidth}
	\vspace{-2.5mm}
    \centering	\includegraphics[width=1\linewidth,trim=0 30 0 0, clip]{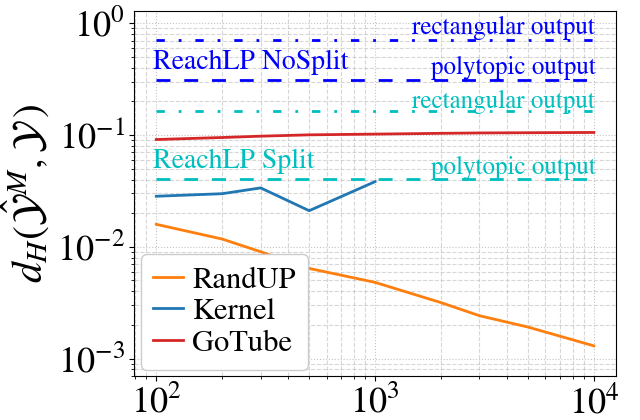}	
\includegraphics[width=1\linewidth]{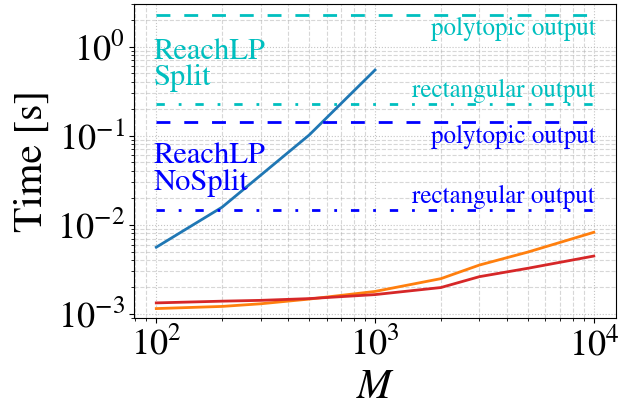}
\vspace{-8mm}
	\caption{Neural network verification analysis in Section \ref{sec:results:verif_closed_nn}: we report the computation time of each algorithm and their averaged Hausdorff distance error (with $\epsilon{=}0$ for $\epsilon$-\randup and \gotube) over $100$ tries when estimating $\Y=\X_4$.}\label{fig:nn_controller:M_vs_dH_time}
	\vspace{-4mm}
	\end{minipage}%
	\end{wrapfigure}
We compare $\epsilon$-\randup with the formal method \reachlp \citep{Everett21_journal}\footnote{Comparisons with \reachsdp \citep{Hu2020}, which is more conservative than \reachlp, show a similar trend.} %
and with two recently-derived sampling-based approaches: the kernel method proposed in \citep{ThorpeL4DC2021} and \gotube \citep{Gruenbacher2021}. We implement \gotube using the $\epsilon$-\randup algorithm where we replace the last convex hull bounding step with an outer-bounding ball. 
As ground-truth, we use the reachable sets from $\epsilon$-\randup with $\epsilon\,{=}\,0$ and $M\,{=}\,10^6$, which is motivated by the asymptotic results from Theorem \ref{thm:asymptotic_convergence} and was previously done in \citep{Everett21_journal}. 
 We refer to Appendix \ref{apdx:exps:nn} for details and present results in Figures \ref{fig:nn_controller:all} and \ref{fig:nn_controller:M_vs_dH_time}.

\textbf{Formal methods} that explicitly bound the output of each layer of the neural network can guarantee that %
their reachable set approximations are always conservative. %
However, obtaining tight approximations with \reachlp requires splitting the input set: a computationally expensive procedure (Fig.\,\ref{fig:nn_controller:M_vs_dH_time}, bottom). 
Figures \ref{fig:nn_controller:all}  and \ref{fig:nn_controller:M_vs_dH_time} show that \reachlp is more conservative than $\epsilon$-\randup even when considering polytopic outputs with eight facets. %
As shown in Figure \ref{fig:nn_controller:all} (right),  %
the conservatism of these methods increases over time. This shows that even when considering small neural networks, verifying safety specifications over long horizons remains an open challenge.

\textbf{Sampling-based} approaches %
do not suffer from the long-horizon conservatism of formal methods. 
This comes at the expense of probabilistic guarantees (that rely on knowledge of the Lipschitz constant of the model), as opposed to deterministic conservatism guarantees. %
$\epsilon$-\randup and \gotube have comparable computation time\footnote{%
Plotting the kernel-based level set estimator in \citep{ThorpeL4DC2021} from $M$ samples requires classifying a dense grid of points. To evaluate the computation time of this method, we only account for the time to classify $M$ new samples.}  and are significantly faster than other approaches. %
$\epsilon$-\randup is significantly more accurate than prior work, especially for larger values of $M$. 
Also, the results from Theorem \ref{thm:conservative_finite_sample} allow for principled hyperparameter selection for $\epsilon$-\randup: given $\epsilon=0.02$, 
sampling $1400$ 
uniformly-distributed inputs on $\partial\X$ is sufficient for the output sets to be conservative with probability at least $1-10^{-4}$ (for $L=1$, see Section \ref{apdx:exps:nn}). 

These experiments show that for short-horizon problems ($5$ steps) with relatively simple network architectures, both \reachlp and $\epsilon$-\randup return accurate reachable set approximations. 
For longer-horizon problems ($9$ steps) with networks of moderate dimensions (which allows using existing methods to pre-compute a Lipschitz constant, see \citep{Fazlyab2019} and Section \ref{appendix:lipschitz_relu}), $\epsilon$-\randup is guaranteed to  efficiently return non-overly-conservative reachable set approximations with high probability. %
Finally, though we do not present such results here, 
the generality of $\epsilon$-\randup allows it to tackle complex model architectures (see \citep{LewEtAl2022} for experiments with longer horizons and more complex networks with uncertain weights) for which no alternative methods exist, albeit without finite-sample accuracy guarantees.

\vspace{-1mm}

\subsection{Application to robust model predictive control}\label{sec:results:mpc}

Finally, we show that $\epsilon$-\randup can be embedded in a robust MPC formulation to reliably control a planar spacecraft system actuated by cold-gas thrusters. Its state at time $t\geq 0$ is denoted as 
$x_t\in\R^6$ and its control inputs are given as $u_t\in\R^3$. 
We use an auxiliary linear feedback controller \citep{LewEtAl2022} and an uncertain linear model $x_{t+1}=f(x_t,u_t,m,F)$ that depends on an uncertain mass $m\in[10,18] \,\textrm{kg}$ (depending on the payload transported by the robot and the current weight of the gas tanks) and an unknown force $F=(F_x,F_y)\in[-0.015,0.015]^2 \,\textrm{N}$ that accounts for the tilt of the table. 
To control the system from an initial state $x_0\in\R^n$ to a goal region $\Xgoal\subset\R^n$ while 
minimizing fuel consumption and 
remaining in a feasible set $\Xsafe$ (i.e., avoiding obstacles and respecting velocity bounds), we consider the following MPC formulation:
\begin{subequations}
\label{eq:full_problem}
\begin{align}
\mathop{\text{min}}_{(\mu,\nu)}
\quad &
\scalebox{0.95}{$\sum_{t=1}^{N}$} (\mu_t-x_{\textrm{goal}})^\top Q(\mu_t-x_{\textrm{goal}})
+
\scalebox{0.95}{$\sum_{t=1}^{N}$}
\nu_t^\top R \nu_t,
\quad 
\textrm{s.t.}
\quad\,
\mu_0=x_0, 
\label{eq:cost:measure}
\\[-1mm]
\textrm{ }
\quad
& 
\mu_{t+1} = f(\mu_t,\nu_t, \bar{m}, \bar{F}),
\ \ 
\nu_t\in\U, 
\ \  
\X_t(\nu) \subset \Xsafe,
\ \ 
\X_N(\nu) \subset \Xgoal,
\ \  
\ {\tiny t\,{=}\,0, \mydots,N\,{-}\,1}
.
\label{eq:robust_constraints_orig}
\end{align}
\end{subequations} 
where $\mu=(\mu_0,\dots,\mu_N)$ %
and $\nu=(\nu_0,\dots,\nu_{N-1})$ %
are optimization variables representing the nominal state and control trajectories, $(\bar{m},\bar{F}_x, \bar{F}_y)=(14,0,0)$ are nominal parameter values, $x_{\textrm{goal}}\in\Xgoal$ is the center of the goal set, and the reachable sets $\X_t(\nu)\subset\R^n$ are defined as
$\X_t(\nu) = %
\{
x_t=f(\cdot,\nu_{t-1},m,F)
\circ\dots\circ
f(x_0,\nu_0,m,F):
\
(m,F)\in[10,18]\times[-0.015,0.015]
\}$.
The numerical implementation %
is described in \citep{LewPavone2020}. 
With a Python implementation, $\epsilon=0.03$, and $M=10^3$, our MPC controller runs at $10$\textrm{Hz} %
which is sufficient for this platform and could be improved, e.g., by parallelizing computations on a GPU. 
We compare with a MPC baseline that does not consider uncertainty over the parameters (i.e., assumes $(m,F)\in\{14\}{\times}\{(0,0)\}$).  As shown in Figure \ref{fig:results:hardware} and in the attached video, this baseline is unsafe and collides with an obstacle. 
In contrast, our reachability-aware controller is recursively feasible, satisfies all constraints, and allows safely reaching the goal. 
These experiments motivate the development of efficient reachability algorithms that can be embedded in generic control frameworks to account for uncertain parameters. %

 \begin{figure}[t]
 \begin{minipage}{.32\linewidth}
 \centering
 \vspace{-7mm}
 \includegraphics[width=1\linewidth]{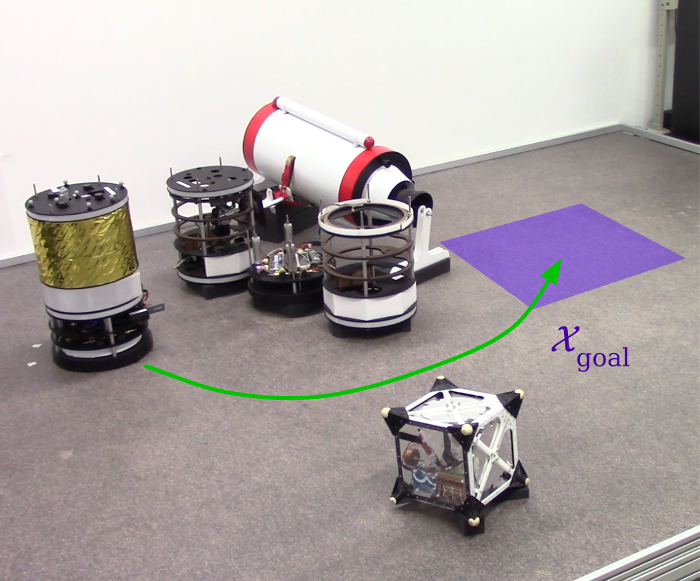}
 \end{minipage}
 \hspace{1mm}
 \begin{minipage}{.26\linewidth}
 \centering
 \vspace{-5mm}
 \includegraphics[width=1\linewidth]{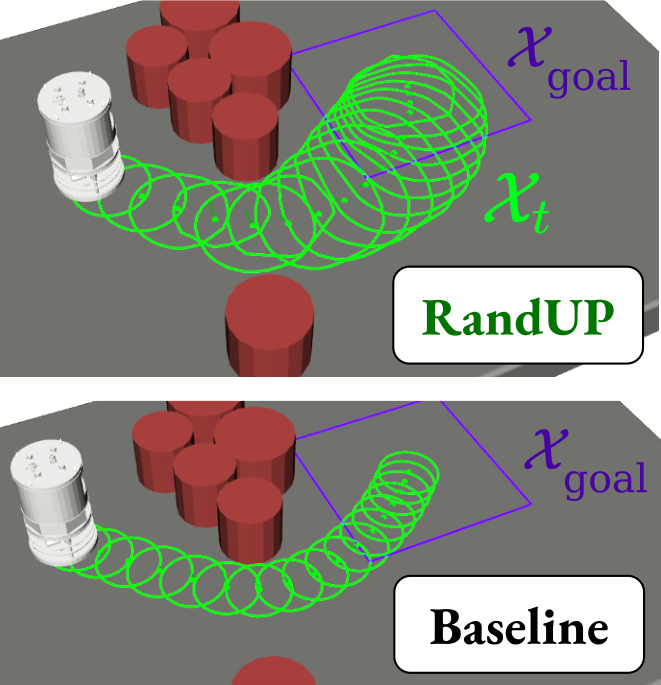}
 \end{minipage}
 \begin{minipage}{.4\linewidth}
 \centering
 \vspace{-5mm}
 \includegraphics[width=0.95\linewidth]{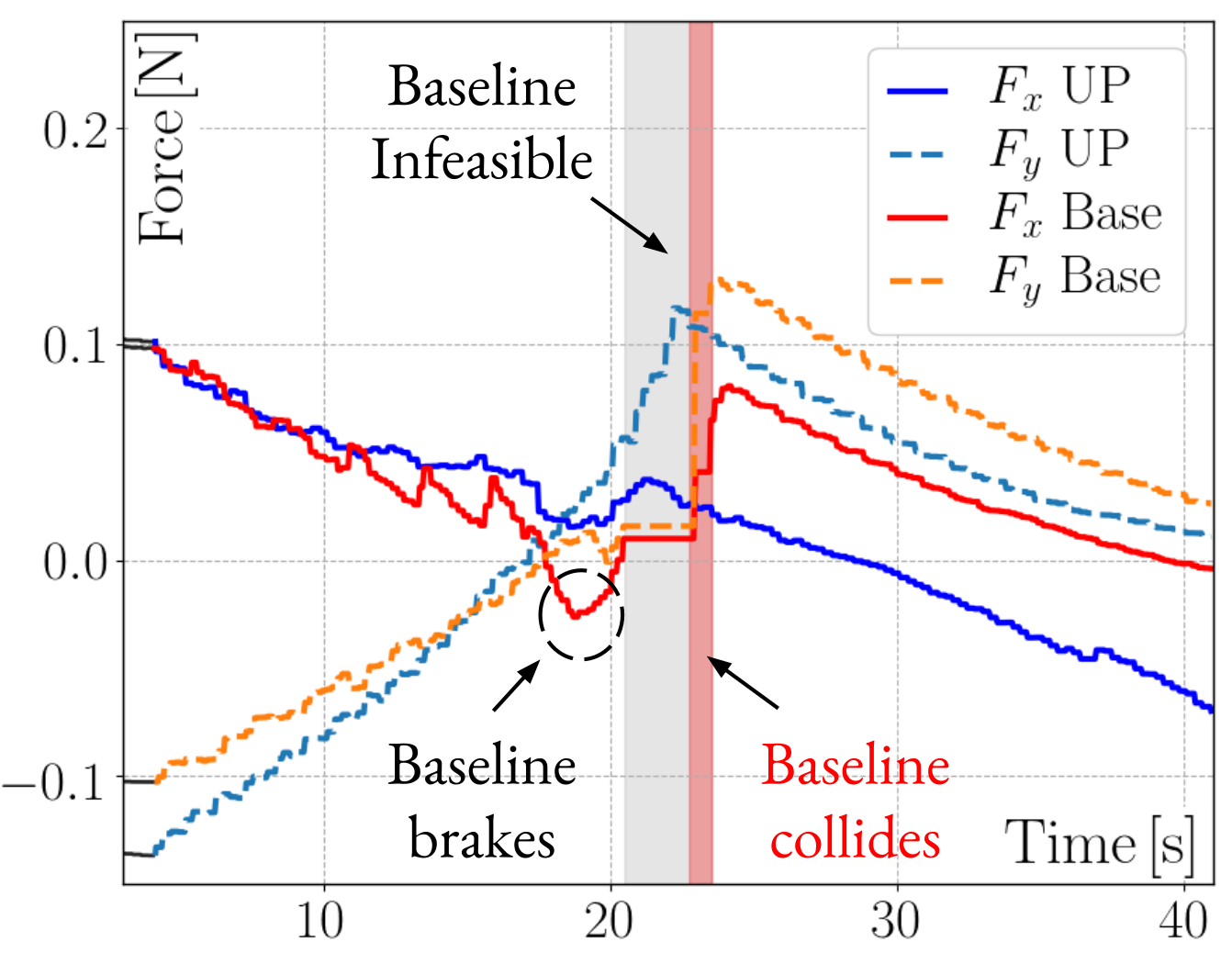}
 \end{minipage}
 \vspace{-3mm}
  \caption{Application of $\epsilon$-\randup to safely control a free-flyer robot in a cluttered environment (left). Using a model predictive controller that does not account for the uncertain dynamics (middle) leads to unsafe behavior, colliding with an obstacle and causing the optimization problem to be infeasible at run-time (right).
}
 \label{fig:results:hardware}
 \vspace{-6mm}
 \end{figure}

\section{Conclusion}
We derived new asymptotic and finite-sample statistical guarantees for $\epsilon$-\randup, a simple yet efficient algorithm for reachability analysis of general systems. We demonstrated its efficacy for a neural network verification task and its applicability to robust model predictive control. 
In future work, we will investigate tighter finite-sample bounds by leveraging further information about the smoothness of the input set boundary $\partial\X$. 
Of practical interest is investigating which sampling distributions enable better sample efficiency, interfacing $\epsilon$-\randup with Lipschitz constant computation methods (e.g., \citep{Fazlyab2019} for neural networks), 
exploring  methods to scale to high-dimensional input spaces, and applying the technique to safety-aware reinforcement learning.
\acks{The authors thank Robin Brown for her helpful feedback and insightful discussions about neural network verification, 
Edward Schmerling for his helpful comments and suggestions, and Adam Thorpe for helpful discussions about kernel methods.  
The NASA University Leadership Initiative (grant \#80NSSC20M0163) provided funds to assist the authors with their research, but this article solely reflects the opinions and conclusions of its authors and not any NASA entity. NVIDIA provided funds to assist the authors with their research. L.J. was supported by the National Science Foundation via grant CBET-2112085.}

\bibliography{ASL_papers,main}

\newpage
\appendix
\section{Formal definitions and random set theory}\label{apdx:random_set_theory}

As a complement to Section \ref{sec:formal_setting}, this section provides a formal description of $\epsilon$-\randup using random set theory. 
Since the set estimator $\hat\Y^M_{\epsilon}$ in \eqref{eq:estimator_eps} is a random variable, describing its measurability properties is important to formally analyze its convergence properties (in an appropriate topology, which we define using the Hausdorff distance). 
In particular, random set theory provides a rigorous framework to characterize the probability distribution of $\hat\Y^M_{\epsilon}$ in Theorems \ref{thm:asymptotic_convergence} and \ref{thm:conservative_finite_sample}.

We denote 
 $\K$ for the family of nonempty compact subsets of $\R^n$, 
$\B(\R^n)$ for the Borel $\sigma$-algebra for the Euclidean topology on $\R^n$ associated to the usual Euclidean norm $\|\cdot\|$, 
$\lambda(\cdot)$ for the Lebesgue measure over $\R^p$, 
$\hull(A)$ for the convex hull of a subset $A\subset\R^n$, 
$\oplus$ for the Minkowski sum, 
$B(x,r)=\{y\in\R^n: \|y-x\|\leq r\}$ for the closed ball of center $x\in\R^n$ and radius $r\geq 0$, 
$\mathring{B}(x,r)$ for the open ball,  
and 
$\partial A$ for the boundary of any $A\subset\R^n$. 

\subsection{Random set theory}

Our analysis hinges on the observation that the set estimator $\hat\Y_\epsilon^M$ is a random compact set, i.e., $\hat\Y_\epsilon^M$ is a random variable taking values in the family of nonempty compact sets $\K$.  
To characterize the accuracy of
our estimator,
we use the \textit{Hausdorff metric}, which is defined for any $A,B\in\K$  in \eqref{eq:metric:Hausdorf} as 
\begin{equation}%
d_\textrm{H}(A,B)
=
\max\big(
\sup_{x\in B}
\Inf_{y\in A}
\|x-y\|, \
\sup_{x\in A}
\Inf_{y\in B}
\|x-y\|
\big).
\end{equation}
This metric induces the \textit{myopic topology} on $\K$ \citep{Molchanov_BookTheoryOfRandomSets2017} with its associated generated Borel $\sigma$-algebra $\B(\K)$. 
$(\K,\B(\K))$ is a measurable space, which motivates the following definition:

\begin{definition}[Random compact set]\label{def:random_compact_set}
Let $(\Omega,\G,\Prob)$ be a probability space. A map $\hat\Y:\Omega\rightarrow\K$ is a random compact set if $\{\omega\in\Omega : \hat\Y(\omega)\in\mathscr{Y}\}\in\G$ for any $\mathscr{Y}\in\B(\K)$.
\end{definition}
Since a random compact set $\hat\Y$ is a random variable with values in $\K$, 
its distribution is characterized by the probability $\Prob(\hat\Y\in\mathscr{Y})$ that it takes values in a measurable subset of compact sets  $\mathscr{Y}\in\B(\K)$. %
Equivalently \citep{Molchanov_BookTheoryOfRandomSets2017}, %
the law of $\hat\Y$ is %
characterized by the capacity functional $T_{\hat\Y}:\K\rightarrow[0,1]$, defined for any $K\in\K$ as 
$
T_{\hat\Y}(K)=
\Prob(\hat\Y \cap K\neq \emptyset)
$. 
This %
functional describes the probability that $\hat\Y$ intersects any compact set $K$. We will analyze this functional to prove the convergence of our set estimator in Theorems \ref{thm:asymptotic_convergence} and \ref{thm:conservative_finite_sample}. 

Our asymptotic convergence result in Theorem \ref{thm:asymptotic_convergence} relies on \citep[Proposition 1.7.23]{Molchanov_BookTheoryOfRandomSets2017}, which provides sufficient conditions for the convergence of random closed sets. 
We restate this result in the particular case of a sequence of random compact sets.

 \begin{thm}[Convergence of Random Sets to a Deterministic Limit \citep{Molchanov_BookTheoryOfRandomSets2017}]\label{thm:conv_randSets_detLim}
 \\
Let $\Y\in\K$ and let $\{\hat\Y^M\}_{M=1}^\infty$ be a sequence of random compact sets. Assume that 
\begin{itemize}[leftmargin=5.5mm]
\item 
For any $K\in\mathcal{K}$ such that $\Y\cap K=\emptyset$,
\begin{equation}\tag{{C1}}\label{eq:C1}
\hspace{-2mm}
\Prob(\hat\Y^M\cap K \neq \emptyset \ \textrm{infinitely often}) 
= 
\Prob\left(\bigcap_{N=1}^\infty\bigcup_{M=N}^\infty
\{\hat\Y^M\cap K \neq \emptyset\}\right) 
=
0.
\end{equation}
\item
For any open subset $G\subset\R^n$ such that $\Y\cap G \neq \emptyset$,
\begin{equation}\tag{{C2}}\label{eq:C2}
\hspace{-2mm}
\Prob(\hat\Y^M\cap G = \emptyset \ \textrm{infinitely often}) 
= 
\Prob\left(\bigcap_{N=1}^\infty\bigcup_{M=N}^\infty
\{\hat\Y^M\cap G = \emptyset\}\right) 
= 
0.
\end{equation}
\end{itemize}
Then, the sequence of random compact sets $\{\Y^M\}_{M=1}^\infty$ 
almost surely converges to $\Y$ (in the myopic topology), i.e., almost surely,
$
d_H(\Y^M,\Y)\rightarrow 0
$ 
 as $M\rightarrow\infty$.
\end{thm}

\subsection{Sampling-based reachability analysis}

As discussed in Section  \ref{sec:formal_setting}, 
$\epsilon$-\randup relies on the choice of three parameters: 
\begin{itemize}\setlength\itemsep{0.5mm}
\item a number of samples $M\in\mathbb{N}$, 
\item a padding constant $\epsilon>0$, 
\item a probability measure $\Prob_\X$ on $(\R^{p},\B(\R^{p}))$ such that $\Prob_\X(\X)=1$. 
\end{itemize}
This algorithm consists of sampling $M$ independent identically-distributed inputs $x_i$ according to $\Prob_\X$, 
of evaluating each output sample $y_i=f(x_i)$, %
and 
of computing the reachable set estimator $\hat\Y_\epsilon^M=\hull\left(\{y_i\}_{i{=}1}^M\right)\oplus B(0,\epsilon)$ as in \eqref{eq:estimator_eps}. 

Formally, let $(\Omega,\G,\Prob)$ be a probability space such that the $x_i$'s are $\G$-measurable independent random variables which laws $\Prob_\X$ satisfy  $\Prob_\X(A)=\Prob(x_i\in A)$ for any $A\in\B(\R^p)$\footnote{For a canonical construction, 
let $
\Omega=\R^p\times\mydots\times\R^p$ ($M$ times), 
$\G=\B(\R^p)\otimes\mydots\otimes\B(\R^p)$,  
$\Prob=\Prob_\X\otimes\mydots\otimes\Prob_\X$ the product measure, and %
$x=(x_1,\mydots,x_M): \Omega\rightarrow\Omega:  \omega\mapsto\omega$. Then, the $x_i$ are independent and have the law $\Prob_\X$.}. %
Then, the $y_i$'s are independent random variables which laws $\Prob_{\Y}$ satisfy %
$\Prob_{\Y}(B)=\Prob(y_i\in B)=\Prob_\X(f^{-1}(B))$ for any $B\in\B(\R^n)$. 
It follows that  
$\hat\Y_\epsilon^M:\Omega\rightarrow\K$ is a random compact set satisfying Definition \ref{def:random_compact_set}\footnote{Compactness of $\X$ and continuity of $f$ guarantee that both $\Y$ and $\hat\Y_\epsilon^M(\omega)$ are compact for any $\omega\in\Omega$. 
We refer to \citep{Molchanov_BookTheoryOfRandomSets2017} and \citep{LewPavone2020} for a proof of measurability of $\hat\Y_\epsilon^M$.}. 
Intuitively, different input samples $x_i(\omega)$  induce different output samples $y_i(\omega)$, resulting in  different approximated compact reachable sets $\hat\Y_\epsilon^M(\omega)\in\K$, where  $\omega\in\Omega$.

\section{Proofs}\label{appendix:proofs}
\subsection{Proof of Theorem \ref{thm:asymptotic_convergence}}\label{apdx:proof:thm:asymptotic_convergence}

We first restate Assumption \ref{assum:XBoundary:posMeasure} and Theorem \ref{thm:asymptotic_convergence} from Section \ref{sec:asymptotic}. 
\\[2mm]
\textbf{Assumption \ref{assum:XBoundary:posMeasure}} $\Prob_\X(\{x\in\X: f(x)\in\mathring{B}(y,r)\})>0$  
for all $y\in\partial\Y$ and all $r>0$.
\\[2mm]
\textbf{Theorem \ref{thm:asymptotic_convergence}}
Let $\bar{\epsilon}\geq 0$ and 
$(\epsilon_M)_{M\in\mathbb{N}}$ be a sequence of padding radii such that $\epsilon_M\geq 0$ for all $M\in\mathbb{N}$ and $\epsilon_M\rightarrow \bar{\epsilon}$ as $M\rightarrow\infty$. 
For any $\epsilon\geq 0$, define the estimator $\hat\Y^M_{\epsilon}=\hull\left(\{y_i\}_{i{=}1}^M\right)\oplus B(0,\epsilon)$.  
Then, under Assumption \ref{assum:XBoundary:posMeasure}, $\Prob$-almost surely, 
as $M\rightarrow\infty$, 
$$
d_H(
\hat\Y_{\epsilon_M}^M, 
\hull(\Y)\oplus B(0,\bar{\epsilon})
)
\mathop{\longrightarrow} 0.
$$
\begin{proof} 
Denote $\hat\Y^M=\hull(\{y_i\}_{i=1}^M)$, so that  $\hat\Y^M_{\epsilon_M}=\hull(\{y_i\}_{i=1}^M)\oplus B(0,\epsilon_M)$. 

To prove that almost surely, the sequence of random compact sets $\{\hat\Y^M_{\epsilon_M}(\omega), M\geq 1\}$ converges to $\hull(\Y)\oplus B(0,\bar{\epsilon})$
 as $M\rightarrow\infty$, we verify the conditions \eqref{eq:C1} and \eqref{eq:C2} of Theorem \ref{thm:conv_randSets_detLim}.

\vspace{2mm}

\textbf{\eqref{eq:C1}}: Let $K\in \mathcal{K}$ satisfy $(\hull(\Y)\oplus B(0,\bar{\epsilon}))\cap K = \emptyset$. 
Then, %
since $K$ and $\hull(\Y)\oplus B(0,\bar{\epsilon})$ are both closed, there exists some $\epsilon>0$ such that $(\hull(\Y)\oplus B(0,\bar{\epsilon}+\epsilon))\cap K = \emptyset$\footnote{More generally, given $A,B\in\K$, then  
$A\cap B=\emptyset$ implies that $A\cap(B\oplus B(0,\epsilon))=\emptyset \text{ for some $\epsilon>0$.}$}. 
 
 Next, since $\epsilon_M\rightarrow\bar{\epsilon}$ as $M\rightarrow\infty$, we have that $\bar{\epsilon}-\epsilon<\epsilon_M<\bar{\epsilon}+\epsilon$ for all $M\geq M_{\epsilon}$. 
 
Since $y_i(\omega)\in\Y$ almost surely for all $i$, $\hat\Y^M\subseteq\hull(\Y)$ almost surely for all $M\geq M_{\epsilon}$. 
Thus, for any $M\geq M_{\epsilon}$, 
$(\hat\Y^M\oplus B(0,\epsilon_M))
\subseteq
(\hull(\Y)\oplus B(0,\epsilon_M))
\subset
(\hull(\Y)\oplus B(0,\bar{\epsilon}+\epsilon))$. 

Combined with $(\hull(\Y)\oplus B(0,\bar{\epsilon}+\epsilon))\cap K = \emptyset$, this implies that $(\hat\Y^M\oplus B(0,\epsilon_M))\cap K=\emptyset$ almost surely for all $M\geq M_{\epsilon}$. 
	
	Thus, $\Prob(\hat\Y^M\oplus B(0,\epsilon_M)\cap K\neq\emptyset)=0$ for all $M\geq M_{\epsilon}$. 
	
	Therefore, $\sum_{M=1}^\infty \Prob(\hat\Y^M\oplus B(0,\epsilon_M)\cap K\neq\emptyset)<\infty$. By the first Borel-Cantelli lemma, we obtain that $\Prob(\hat\Y^M\oplus B(0,\epsilon_M)\cap K \neq \emptyset \ \ i.o.) = 0$. This concludes \eqref{eq:C1}. 
	
\vspace{2mm}

\textbf{\eqref{eq:C2}}: Let $G\subset\R^n$ be an open set satisfying $(\hull(\Y)\oplus B(0,\bar{\epsilon}))\cap G \neq \emptyset$. Equivalently, let $G\subset\R^n$ satisfy $\hull(\Y)\cap (G \oplus B(0,\bar{\epsilon}))\neq \emptyset$. We wish to prove that $\Prob(\hat\Y^M_{\epsilon_M}\cap G  = \emptyset \ i.o.) = 0$.

Let $G_\partial^1,G_\partial^2\subset\R^n$ be two arbitrary boundary-intersecting open sets such that $\partial\Y\cap (G_\partial^1\oplus B(0,\bar{\epsilon})) \neq \emptyset$ and $\partial\Y\cap (G_\partial^2\oplus B(0,\bar{\epsilon})) \neq \emptyset$. We will select specific sets $G_\partial^j$ as a function of $G$ later in the proof.

 \begin{figure}[t]
    \centering
	\includegraphics[width=0.25\linewidth]{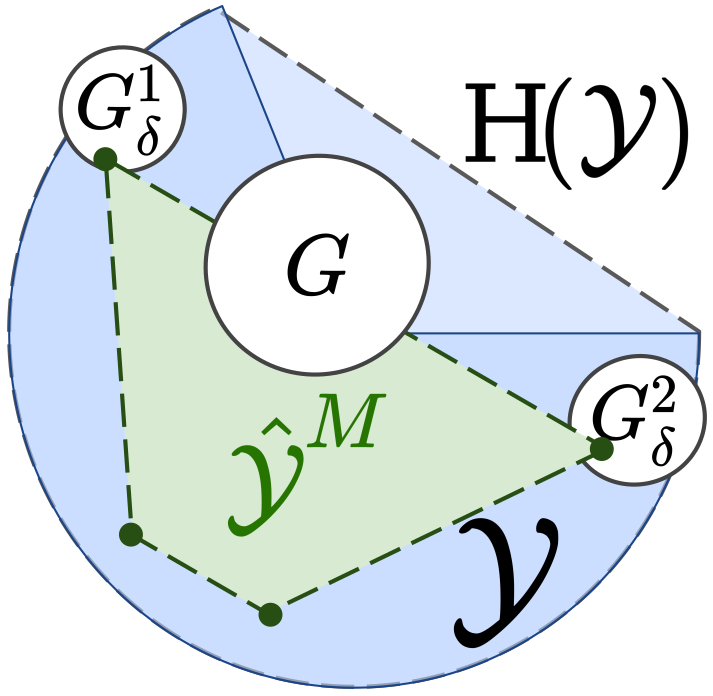}
	\caption{Particular case with $\bar\epsilon=0$: for any set $G\subset\R^n$ that intersects $\hull(\Y)$ (i.e.,  $\hull(\Y)\cap G \neq \emptyset$), there exists two boundary-intersecting sets $G_\partial^1,G_\partial^2\subset\R^n$ (i.e. $\partial\hull(\Y)\cap G_\partial^j \neq \emptyset$) such that the convex hull of two points $y_1\in G_\partial^1$ and $y_2\in G_\partial^2$ intersects $G$ (i.e., $\hull(\{y_1,y_2\})\cap G\neq \emptyset$). With this fact, we quantify the probability that $\hat\Y^M$ has at least one vertex in $G_\partial^1$ and one in $G_\partial^2$, which guarantees that $\hat\Y^M$ intersects $G$.}\label{fig:proof:Gdelta}
	\end{figure}

To prove \eqref{eq:C2}, we proceed in three steps: 
(C2.1) we show that sampling outputs $y_i$ within the boundary sets $G_\partial^j$ occurs infinitely often (i.o.); 
(C2.2) we derive a sufficient conditions for \eqref{eq:C2} using the growth property $\hat\Y^M_\epsilon\subseteq\hat\Y^{M+1}_\epsilon$ for any fixed $\epsilon\geq 0$; 
(C2.3) we relate the probability of sampling $y_i$ within two well-chosen boundary-intersecting sets $G_\partial^1,G_\partial^2\subset\R^n$ with the probability of $\hat\Y^M_{\epsilon_M}$ to intersect $G\oplus B(0,\bar{\epsilon})$.

\textbf{(C2.1) }
Consider the events $A_{2i}=\{\omega\in\Omega\,|\,
y_{2i+1}\in (G_\partial^1\oplus B(0,\bar{\epsilon})),
y_{2(i+1)}\in (G_\partial^2\oplus B(0,\bar{\epsilon}))\}$, where $i\in\mathbb{N}$.
\\
Since the inputs $x_i$ are sampled independently, the outputs $y_i$ are independent. Thus, 
the events $\{A_{2i}\}_{i=0}^\infty$ 
are independent and  
$\Prob(A_{2i})=
\Prob(y_{2i+1}\in (G_\partial^1\oplus B(0,\bar{\epsilon})))
\Prob(y_{2i+2}\in (G_\partial^2\oplus B(0,\bar{\epsilon})))$. 
Since $\partial\Y\cap (G_\partial^1\oplus B(0,\bar{\epsilon})) \neq \emptyset$ for $j=1,2$, we use Assumption \ref{assum:XBoundary:posMeasure} to obtain that $\Prob(A_{2i}) \geq \delta$ for some $\delta>0$ that depends on the choice of $\Prob_\X$, $G_\partial^j$, and $\bar\epsilon$. 
Thus, 
$\sum_{i=0}^\infty\Prob(A_{2i})=\infty$. Therefore, we apply the second Borel-Cantelli lemma and obtain that 
$
\Prob\big(\bigcup_{N=0}^\infty\bigcap_{M=N}^\infty A_{2M}\big) = 1$.
\\
Next, since $\bigcap_{N=0}^\infty A_n \subseteq A_0$ for any events $A_i$, 
$\Prob\big(\bigcup_{N=0}^\infty\bigcap_{M=N}^\infty A_{2M}\big) 
 \leq 
 \Prob\big(\bigcup_{M=0}^\infty A_{2M}\big) $. 
 Combining the last two results, 
\begin{equation}\label{eq:prob_union_xm_in_G_equals_1}
\Prob
\bigg(
	\bigcup_{M=0}^\infty 
(y_{2M+1}\in (G_\partial^1\oplus B(0,\bar{\epsilon})))
\cap 
(y_{2M+2}\in (G_\partial^2\oplus B(0,\bar{\epsilon})))
\bigg) = 1.
\end{equation}
(C2.2) To proceed with this step, we first observe the following. 
\\
$G$ satisfies 
$\hull(\Y)\cap (G\oplus B(0,\bar{\epsilon}))\neq \emptyset$. Therefore, since $\hull(\Y)$ is compact and $G$ is open, there exists some $\epsilon>0$ such that  
$\hull(\Y)\cap (G\oplus B(0,\bar{\epsilon}-\epsilon))\neq \emptyset$.\footnote{More generally, given $A\subset\K$ and $B\subset\R^n$ an open set, then 
$
A\cap B \neq\emptyset\implies A\cap(B\ominus B(0,\epsilon))\neq\emptyset \text{ for some $\epsilon>0$.}$}
\\
Since $\epsilon\rightarrow\bar{\epsilon}$ as $M\rightarrow\infty$, 
there exists some $M_\epsilon\in\mathbb{N}$ such that $\bar{\epsilon}-\epsilon<\epsilon_M<\bar{\epsilon}+\epsilon$ for all $M\geq M_\epsilon$. 
\\
Second, we rewrite \eqref{eq:C2} as
\begin{align}
\Prob(\hat\Y^M_{\epsilon_M}\cap G = \emptyset \ \ i.o.)
=
\Prob\bigg(\bigcap_{N=1}^\infty\bigcup_{M=N}^\infty \hat\Y^M_{\epsilon_M}\cap G = \emptyset\bigg)
\nonumber
=
1 - \Prob\bigg(\bigcup_{N=1}^\infty\bigcap_{M=N}^\infty \hat\Y^M_{\epsilon_M}\cap G \neq \emptyset\bigg).
\end{align}
Next, note that 
\begin{align*}
\Prob\bigg(\bigcup_{N=1}^\infty\bigcap_{M=N}^\infty \hat\Y^M_{\epsilon_M}\cap G \neq \emptyset\bigg)
&\geq 
\Prob\bigg(\bigcup_{N=M_\epsilon}^\infty\bigcap_{M=N}^\infty \hat\Y^M_{\epsilon_M}\cap G\neq \emptyset\bigg)
\geq 
\Prob\bigg(
\bigcup_{N=M_\epsilon}^\infty
\bigcap_{M=N}^\infty
\hat\Y^M_{\bar{\epsilon}-\epsilon}\cap G \neq \emptyset\bigg).
\end{align*}
The second inequality holds since $\hat\Y^M_{\bar{\epsilon}-\epsilon}\subset \hat\Y^M_{\epsilon_M}$ if $M\geq M_\epsilon$. %
\\
Next, 
since $\hat\Y^M_{\bar{\epsilon}-\epsilon}\subseteq\hat\Y^{M+1}_{\bar{\epsilon}-\epsilon}$ for any $M\in\mathbb{N}$, we have that 
$\{\omega\in\Omega\, |\, \bigcap_{M=N}^\infty \hat\Y^M_{\bar{\epsilon}-\epsilon}\cap G \neq \emptyset\}
=
\{\omega\in\Omega \, |\, \hat\Y^N_{\bar{\epsilon}-\epsilon}\cap G \neq \emptyset\}$. Thus, 
\begin{equation}\nonumber
\Prob\bigg(\bigcup_{N=M_{\epsilon}}^\infty\bigcap_{M=N}^\infty \hat\Y^M_{\bar{\epsilon}-\epsilon}\cap G \neq \emptyset\bigg)=\Prob\bigg(\bigcup_{M=M_\epsilon}^\infty\hat\Y^M_{\bar{\epsilon}-\epsilon}\cap G \neq \emptyset\bigg).
\end{equation}
Combining the last three results, we obtain the following sufficient condition for \eqref{eq:C2}
\begin{equation}\label{eq:io_union_iff}
\Prob\bigg(\bigcup_{M=M_\epsilon}^\infty\hat\Y^M_{\bar{\epsilon}-\epsilon}\cap G \neq \emptyset\bigg) = 1
	\implies
	\Prob(\hat\Y^M_{\epsilon_M}\cap G = \emptyset \ \ i.o.) = 0 
.
\end{equation}
(C2.3) 
Finally, we combine \eqref{eq:prob_union_xm_in_G_equals_1} and \eqref{eq:io_union_iff} as follows. 
For $M\geq 0$, we have that
\begin{align*}
&\{
\omega\in\Omega: 
(y_{2M+1}(\omega)\in (G_\partial^1\oplus B(0,\bar{\epsilon})))
\cap 
(y_{2M+2}(\omega)\in (G_\partial^2\oplus B(0,\bar{\epsilon})))
\}
\\
&\qquad\subseteq
\{\omega\in\Omega \,|\, 
(\hat\Y^{2M+2}_{\bar{\epsilon}-\epsilon}(\omega)\cap (G_\partial^1\oplus B(0,\bar{\epsilon}))\neq\emptyset)
\cap
(\hat\Y^{2M+2}_{\bar{\epsilon}-\epsilon}(\omega)\cap (G_\partial^2\oplus B(0,\bar{\epsilon}))\neq\emptyset)
\} 
\end{align*}
Thus, 
\begin{align*}
&\bigcup_{M=0}^\infty
\{\omega\in\Omega: (y_{2M+1}(\omega)\in (G_\partial^1\oplus B(0,\bar{\epsilon})))\cap (y_{2M+2}(\omega)\in (G_\partial^2\oplus B(0,\bar{\epsilon})))\}
\\
&\hspace{5mm}\subseteq
\bigcup_{M=2}^\infty
\{\omega\in\Omega \,|\,  
(\hat\Y^M_{\bar{\epsilon}-\epsilon}(\omega)\cap G_\partial^1\neq\emptyset)
\cap
(\hat\Y^M_{\bar{\epsilon}-\epsilon}(\omega)\cap G_\partial^2\neq\emptyset)
\} 
\\
&\hspace{5mm}\subseteq
\bigcup_{M=1}^\infty
\{\omega\in\Omega \,|\, 
(\hat\Y^M_{\bar{\epsilon}-\epsilon}(\omega)\cap G_\partial^1\neq\emptyset)
\cap
(\hat\Y^M_{\bar{\epsilon}-\epsilon}(\omega)\cap G_\partial^2\neq\emptyset)
\}.
\end{align*}
From \eqref{eq:prob_union_xm_in_G_equals_1}, the first event holds with probability one. Therefore, by the above, 
$$
\Prob\left(
\bigcup_{M=M_\epsilon}^\infty
(\hat\Y^M_{\bar{\epsilon}-\epsilon}\cap G_\partial^1\neq\emptyset)
\cap
(\hat\Y^M_{\bar{\epsilon}-\epsilon}\cap G_\partial^2\neq\emptyset)
\right)=1.
$$
Given the right choice of $G_\delta^1,G_\delta^2$, by convexity of $\hull(\Y)$\footnote{$\hull(\Y)\oplus B(0,\bar\epsilon)$ is convex, so that any $y\in G\cap (\hull(\Y)\oplus B(0,\bar\epsilon))\neq\emptyset$ lies on a line passing through two extreme points $z_1^y,z_2^y$ of $\hull(\Y)\oplus B(0,\bar\epsilon)$.  
For some $\tilde\epsilon>0$, let $G_\partial^1=\mathring{B}_\partial^1(z_1^y,\tilde\epsilon)$ and $G_\partial^2=\mathring{B}_\partial^2(z_2^y,\tilde\epsilon)$ (note that $G_\partial^1,G_\partial^2$ intersect $\hull(\Y)\oplus B(0,\bar\epsilon)$). 
Then,  
by choosing $\tilde\epsilon$ small enough, 
since  $\hat\Y^M_{\bar{\epsilon}-\epsilon}$ is convex,  the inequality follows since  $\hat\Y^M_{\bar{\epsilon}-\epsilon}$ necessarily intersects $G$ if it intersects $G_\partial^1$ and $G_\partial^2$. } (see also Figure \ref{fig:proof:Gdelta}), we have that
\begin{align*}\label{eq:prob_YM_cap_G_neq_empty_is_1}
\Prob\left(\bigcup_{M=M_\epsilon}^\infty
\hat\Y^M_{\bar{\epsilon}-\epsilon}\cap G
\neq\emptyset\right)\geq
\Prob\left(
\bigcup_{M=M_\epsilon}^\infty
(\hat\Y^M_{\bar{\epsilon}-\epsilon}\cap G_\partial^1\neq\emptyset)
\cap
(\hat\Y^M_{\bar{\epsilon}-\epsilon}\cap G_\partial^2\neq\emptyset)
\right)
=1.
\end{align*}
Therefore, 
$\Prob\left(\bigcup_{M=M_\epsilon}^\infty
\hat\Y^M_{\bar{\epsilon}-\epsilon}\cap G
\neq\emptyset\right)=1$. 
Combining this result with \eqref{eq:io_union_iff}, we obtain that $\Prob(\hat\Y^M_{\epsilon_M}\cap G = \emptyset \ i.o.) = 0$. 
This conludes the proof of \eqref{eq:C2}.

\vspace{2mm}

By Theorem \ref{thm:conv_randSets_detLim}, we conclude that almost surely, the sequence $\{\hat\Y^M_{\epsilon_M}(\omega),m\geq 1\}$  converges to $\hull(\Y)\oplus B(0,\bar{\epsilon})$ as $M\rightarrow\infty$. 
This concludes the proof of Theorem \ref{thm:asymptotic_convergence}.
\end{proof}

\subsection{Proof of Theorem \ref{thm:conservative_finite_sample}}\label{apdx:proof:thm:conservative_finite_sample}

We start with four intermediate results and then prove Theorem \ref{thm:conservative_finite_sample}. We use the notations introduced in Section \ref{apdx:random_set_theory} throughout this section.

\begin{lemma}%
\label{lem:dX_packing_nb_Lip}
Under Assumption \ref{assum:f:lipschitz} \rev{and assuming that $\partial\Y\subseteq f(\partial\X)$,} 
$$D(\partial\Y,\epsilon)\leq D(\partial\X,\epsilon/L).$$ 
\end{lemma}

\begin{lemma}%
\label{lem:Xball_coverage_prob} Under Assumption \ref{assum:f:lipschitz}, for any $x\in\R^p$, 
$y=f(x)$, and 
any $\delta>0$, %
$$
\Prob_\Y\big(B(y,\delta)\big)
\geq \Prob_\X\big(B(x,\epsilon)\big)
\quad \text{for all }\ \epsilon\in [0,\delta/L].
$$
\end{lemma}
 
\begin{lemma}%
\label{lem:Xhulleps_conservative_geom}
Let $\epsilon>0$ and define 
$$
Y_\epsilon^M
=
\underset{i=1}{\overset{M}{\bigcup}} \,
B(y_i,\epsilon),
\qquad
\pi(\partial\Y,Y_\epsilon^M)
=
\sup_{y\in\partial\Y}\Prob(\{y\} \cap Y_\epsilon^M = \emptyset).
$$  
Then, 
$$
\Prob(\partial\Y\subset Y_{2\epsilon}^M) \geq 1 - D(\partial\Y,\epsilon)
\pi(\partial\Y,Y_\epsilon^M).
$$
\end{lemma}

\begin{lemma}%
\label{lem:Hausdorff_hulls}
Let $\epsilon\geq 0$ and let $Y\in\K$ be such that    
$\partial\Y\subseteq Y\oplus B(0,\epsilon)$ and $Y\subseteq \Y$. 
Then, $d_H(
\hull(Y), 
\hull(\Y)
)\leq \epsilon$.
\end{lemma}

Lemma \ref{lem:Xhulleps_conservative_geom} is the key to deriving Theorem \ref{thm:conservative_finite_sample}. 
It is first derived in \citep{Dumbgen1996} in the convex problem setting. Notably, Lemma \ref{lem:Xhulleps_conservative_geom} does not require the convexity of $\Y$.
\\

\begin{proof}\textbf{of Lemma \ref{lem:dX_packing_nb_Lip}.} 
First, note that 
since $\X$ is compact and $\partial\X\subseteq\X$, $\partial\X$ is compact (note that the boundary is always closed). Any compact set has a finite covering number, thus $D(\partial\X, \epsilon)$ is finite for any finite $\epsilon>0$. 
Second,   
given any $\delta>0$, 
any $x\in\X$, and $y=f(x)$,
\begin{equation}\label{eq:lipschitz_ball_conservative}
f(B(x,\epsilon)) \subseteq B(y,\delta)
\quad
\forall \epsilon\in[0,\delta/L],
\end{equation}
where $L$ is the Lipschitz constant in Assumption \ref{assum:f:lipschitz}. %
Indeed, let $\tilde{x}\in B(x,\epsilon)$, so that $\|\tilde{x}-x\|\leq\epsilon$. Then, $
\|f(\tilde{x})-y\|
\leq L\,\|\tilde{x}-x\|\leq L\epsilon\leq\delta$ where the last inequality holds given that $\epsilon \leq \delta/L$.

Next, let $F_{\partial\X}=\{x_i\}_{i=1}^{|F_{\partial\X}|}\subseteq\partial\X$ be a minimum $(\epsilon/L)$-covering for $\partial\X$, so that $|F_{\partial\X}|=D(\partial\X,\epsilon/L)$ and 
 for any $x\in\partial\X$, there exists $x_i\in F_{\partial\X}$ such that $\|x-x_i\|\leq \epsilon/L$. %
Then, $F_{\partial\Y}=\{f(x_i)\, |\,x_i\in F_{\partial\X} \}$ is an $\epsilon$-covering for $\partial\Y$. 
	Indeed, for any $y\in\partial\Y$, there exists some $x\in\partial\X$ such that $y=f(x)$, and there exists some $x_i\in F_{\partial\X}$ such that $\|x-x_i\|\leq \epsilon/L$. Therefore, 
	$$\underset{y_i\in F_{\partial\Y}}{\sup} \|y-y_i\| = \underset{x_i\in F_{\partial\X}}{\sup} \|f(x)-f(x_i)\| 
\leq \underset{x_i\in F_{\partial\X}}{\sup} L\|x-x_i\| \leq \epsilon.$$ 

Therefore, since $F_{\partial\Y}$ is an $\epsilon$-covering for $\partial\Y$ and $|F_{\partial\Y}|=|F_{\partial\X}|=D(\partial\X,\epsilon/L)$, we obtain that 
$D(\partial\Y,\epsilon)\leq |F_{\partial\Y}| = D(\partial\X,\epsilon/L)$ which concludes this proof.
\end{proof}

\begin{proof}\textbf{of Lemma \ref{lem:Xball_coverage_prob}.}
From \eqref{eq:lipschitz_ball_conservative}, 
given any $\delta>0$, 
 $x\in\R^p$, and 
 $y=f(x)$, 
$
f(B(x,\epsilon)) \subset B(y,\delta)$ for any $\epsilon\in[0,\delta/L]$.
Thus, \vspace{-3mm}
\begin{align*}
\Prob_{\Y}(B(y,\delta))
&=
\Prob_\X\Big(f^{-1}(B(y,\delta))\Big)
\geq
\Prob_\X\left(
B(x,\epsilon)
\right)
\quad
\forall \epsilon\in[0,\delta/L],
\end{align*}
where the last inequality holds since $B(x,\epsilon) \subseteq f^{-1}(B(y,\delta))$,
which concludes this proof.
\end{proof}

\begin{proof}\textbf{of Lemma \ref{lem:Xhulleps_conservative_geom}.} 
Let $F_{\partial\Y}=\{y_i\}_{i=1}^{|F_{\partial\Y}|}\subseteq\partial\Y$ be a minimum $\epsilon$-covering for $\partial\Y$, so that $|F_{\partial\Y}|=D(\partial\Y,\epsilon)$ and 
 for any $y\in\partial\Y$, there exists $y_i\in F_{\partial\Y}$ such that $\|y-y_i\|\leq \epsilon$. 
 Let $B(F_{\partial\Y},\epsilon)=\bigcup_{y_i\in F_{\partial\Y}} B(y_i,\epsilon)$.

By construction, $\partial\Y\subset B(F_{\partial\Y},\epsilon)$,
so that 
$
\partial\Y\nsubseteq Y_{2\epsilon}^M\implies
B(F_{\partial\Y},\epsilon)\nsubseteq Y_{2\epsilon}^M. 
$
Thus, 
\begin{align*}
\Prob(\partial\Y\nsubseteq Y_{2\epsilon}^M)
&\leq 
\Prob(B(F_{\partial\Y},\epsilon)\nsubseteq Y_{2\epsilon}^M)
=
\Prob(F_{\partial\Y}\nsubseteq Y_\epsilon^M)
\\
&=
\Prob\bigg(\bigcup_{y_i\in F_{\partial\Y}} y_i\notin Y_\epsilon^M\bigg)
\\
&\leq 
\sum_{y_i\in F_{\partial\Y}}\Prob(y_i\notin Y_\epsilon^M)
=
\sum_{y_i\in F_{\partial\Y}}
\Prob(\{y_i\}\cap Y_\epsilon^M = \emptyset)
\\
&\leq 
|F_{\partial\Y}| \cdot
\sup_{y\in\partial\Y}\Prob(\{y\}\cap Y_{\epsilon}^M = \emptyset)
\\
&\leq 
D(\partial\Y,\epsilon)\pi(\partial\Y,Y_{\epsilon}^M).
\end{align*}
The conclusion follows.
\end{proof}

\begin{proof}\textbf{of Lemma \ref{lem:Hausdorff_hulls}.} 
$\partial\Y\subseteq Y\oplus B(0,\epsilon)$ implies that $\hull(\Y)=\hull(\partial\Y)\subseteq \hull(Y\oplus B(0,\epsilon))=\hull(Y)\oplus B(0,\epsilon)$.

$Y\subseteq \Y$ implies %
that 
$\hull(Y)\subseteq \hull(\Y)\subset\hull(\Y)\oplus B(0,\epsilon)$.%

Together, 
$\hull(\Y)\subset\hull(Y)\oplus B(0,\epsilon)$
and $\hull(Y)\subset\hull(\Y)\oplus B(0,\epsilon)$
imply that $d_H(
\hull(Y), 
\hull(\Y)
)\leq \epsilon$, see \citep{Schneider2014}.
\end{proof}

With these results, we prove Theorem \ref{thm:conservative_finite_sample} below. We first restate it for better readability.
\\[2mm]
\textbf{Theorem \ref{thm:conservative_finite_sample}} 
 \textit{Define the probability threshold
$
\delta_M= 
D(\partial\X,\epsilon/(2L))\left(1 -
	\Lambda_{\epsilon}^{L}
\right)^M
$ and the estimator $\hat\Y^M=\hull\left(\{y_i\}_{i{=}1}^M\right)$.
Then, under Assumptions \ref{assum:f:lipschitz} and \ref{assum:sampling_density} 
\rev{and assuming that $\partial\Y\subseteq f(\partial\X)$,}  
$$
\Prob(
d_H(
\hat\Y^M, 
\hull(\Y)
)\leq \epsilon
)\geq 1-\delta_M
\qquad\text{and}\qquad
\Prob(\Y\subseteq\hat\Y_\epsilon^M)\geq 1-\delta_M.
$$}
\rev{\hspace{-1.5mm}\textbf{Remark:} the assumption $\partial\Y\subseteq f(\partial\X)$ holds if the reachability map $f$ is open, e.g., if it is a submersion (its differential is surjective). 
If $\partial\Y\nsubseteq f(\partial\X)$, then one could modify Theorem \ref{thm:conservative_finite_sample} by replacing Assumption \ref{assum:sampling_density} with \textit{``Given $\epsilon,L\,{>}\,0$, there exists $\Lambda_{\epsilon}^{L}\,{>}\,0$ such that  $\Prob_\X\left(B\left(x,\frac{\epsilon}{2L}\right)\right)\,{\geq}\, \Lambda_{\epsilon}^{L}$ for all  $x\in\X$''} (i.e., one should sample over the entire set $\X$ and not only along the boundary) and by defining $\delta_M=D(\X,\epsilon/(2L))(1 -	\Lambda_{\epsilon}^{L})^M$.}
\\

\begin{proof}%
As in Lemma \ref{lem:Xhulleps_conservative_geom}, define $Y_\epsilon^M
	=
	\bigcup_{i=1}^M B(y_i,\epsilon)$, 
	$Y^M=\{y_i\}_{i=1}^M$, and 
$$
\pi(\partial\Y,Y_\epsilon^M)
=
\underset{y\in\partial\Y}{\sup}\Prob(\{y\} \cap Y_\epsilon^M = \emptyset)
=
\underset{y\in\partial\Y}{\sup}\Prob(B(y,\epsilon) \cap Y^M = \emptyset),
$$
which corresponds to the worst probability over $y\in\partial\Y$ of not sampling some $y_i$ that is $\epsilon$-close to $y$. 
First, we derive a bound for $\pi(\partial\Y,Y_\epsilon^M)$. 
Using the fact that the samples $y_i$ are i.i.d., 
\begin{align*}
\pi(\partial\Y,Y_\epsilon^M) 
&=
\sup_{y\in\partial\Y}\Prob(B(y,\epsilon) \cap Y^M = \emptyset)
\\
&= \sup_{y\in\partial\Y}\Prob\left(\bigcap_{i=1}^M
(y_i\notin B(y,\epsilon))\right)
\\
&= \left(1 -
	\inf_{y\in\partial\Y}\Prob(y_i\in B(y,\epsilon)) 
\right)^M
\\
&=\left(1 -
	\inf_{y\in\partial\Y}\Prob_\Y(B(y,\epsilon)) 
\right)^M.
\end{align*}
From Lemma \ref{lem:Xball_coverage_prob}, 
for any $x\in\R^p$, 
$y=f(x)$, and 
$\epsilon>0$, 
$\Prob_\Y(B(y,\epsilon))
\geq \Prob_\X(B(x,\epsilon/L))$. 
Since for all $y\in\partial\Y$, there exists $x\in\partial\X$ such that $y=f(x)$, we combine the two previous results to obtain 
\begin{align*}
\pi(\partial\Y,Y_\epsilon^M) 
\leq 
\left(1 -
	\inf_{x\in\partial\X}\Prob_\X(B(x,\epsilon/L))
\right)^M.
\end{align*}
In particular, using Assumption \ref{assum:sampling_density}, 
\begin{align*}
\pi(\partial\Y,Y_{\epsilon/2}^M) 
\leq 
\left(1 -
	\inf_{x\in\partial\X}\Prob_\X(B(x,\epsilon/(2L)))
\right)^M
\leq 
\left(1 -
	\Lambda_{\epsilon}^{L}
\right)^M.
\end{align*}
To complete the proof of Theorem \ref{thm:conservative_finite_sample}, we use Lemma \ref{lem:Xhulleps_conservative_geom} which states that 
$$
\Prob(\partial\Y\subseteq Y_\epsilon^M) \geq 1 - D(\partial\Y,\epsilon/2)
\pi(\partial\Y,Y_{\epsilon/2}^M)
.
$$
Using Lemma \ref{lem:dX_packing_nb_Lip},  
$D(\partial\Y,\epsilon/2)\leq D(\partial\X,\epsilon/(2L))$. 
\\
By Lemma \ref{lem:Hausdorff_hulls}, 
 if $\partial\Y\subseteq Y_\epsilon^M$
and $Y^M\subseteq\Y$\footnote{Since $\Prob_\X(\X)=1$, we have that $\Prob_\Y(\Y)=1$, so that $Y^M\subseteq\Y$ with probability one.}, then 
$d_H(
\hull(Y^M), 
\hull(\Y)
)\leq \epsilon
$  
.
\\
Therefore, with $\hat\Y^M=\hull(Y^M)$ and Assumption \ref{assum:sampling_density}, combining the last inequalities,
\begin{align*}
\Prob(
d_H(
\hat\Y^M, 
\hull(\Y)
)\leq \epsilon
)
&\geq 
\Prob\left(\partial\Y\subseteq Y_\epsilon^M\right)
\\
&\geq 
1 - D(\partial\Y,\epsilon/2)
\pi(\partial\Y,Y_{\epsilon/2}^M)
\\
&\geq
1 - D(\partial\X,\epsilon/(2L))
\left(1 -
	\Lambda_{\epsilon}^{L}
\right)^M.
\end{align*}
If $d_H(
\hat\Y^M, 
\hull(\Y)
)\leq \epsilon$, 
then 
$\Y\subseteq\hull(\Y)\subseteq \hat\Y^M\oplus B(0,\epsilon)=\hat\Y_\epsilon^M$. 
The conclusion follows.
\end{proof}

\subsection{Proof of Corollary \ref{cor:conservative_finite_sample}}\label{apdx:proof:cor:conservative_finite_sample}

We start with the following preliminary result:

\begin{lemma}%
\label{lem:X_vol_covrg_eps_ball:improved_caps}
Assume that $\X^\comp$ is $r$-convex  (Assumption \ref{assum:Theta:r_convex}).
Define 
$\vec{r}=(r,0,\dots,0)\in\R^p$. 
Then, for any $\epsilon\geq 0$,  
$$\underset{x\in\partial\X}{\inf}\lambda(\X \cap B(x,\epsilon))
\geq 
\lambda\big(
B(0,\epsilon) 
\cap B(\vec{r},r)
\big)
.$$
\end{lemma}

\begin{proof}\textbf{of Lemma \ref{lem:X_vol_covrg_eps_ball:improved_caps}.}
By Assumption \ref{assum:Theta:r_convex}, $\X^\comp$ is $r$-convex. Thus, 
 for any $x\in\partial\X$, there exists some $\tilde{x}\in\Int(\X)$ and a (closed) ball $B(\tilde{x},r)$ such that $x\in B(\tilde{x},r)$ and $B(\tilde{x},r)\subseteq\X$. 
 Since $x\in\partial\X$ and $B(\tilde{x},r)\subseteq\X$, 
 $\|x-\tilde{x}\|=r$. %

Let $\epsilon\geq 0$ and $\vec{r}=(r,0,\dots,0)\in\R^p$. Then, by translational and rotational invariance of the Lebesgue measure,
$$
 \lambda(\X \cap B(x,\epsilon))
\geq 
\lambda\big(
B(\tilde{x},r) 
\cap 
B(x,\epsilon)
\big)
=
\lambda\big(
B(x-\tilde{x},r) 
\cap 
B(0,\epsilon)
\big)
=
\lambda\big(
B(\vec{r},r) 
\cap 
B(0,\epsilon)
\big)
$$
As this holds for $x\in\partial\X$, %
we obtain that 
$\underset{x\in\partial\X}{\inf}\lambda(\X \cap B(x,\epsilon))
\geq 
\lambda\big(
B(0,\epsilon) 
\cap B(\vec{r},r)
\big)
$.
\end{proof}

Then, we restate Corollary  \ref{cor:conservative_finite_sample} and prove it below. 
\\[2mm]
\textbf{Corollary \ref{cor:conservative_finite_sample}} 
\textit{Define the estimator $\hat\Y^M=\hull\left(\{y_i\}_{i{=}1}^M\right)$, 
 the offset vector 
 $\vec{r}=(r,0,\dots,0)\in\R^p$, 
 the volume $\Lambda_{\epsilon}^{r,L}=\lambda\big(
 B(0,\epsilon/(2L)) 
 \cap B(\vec{r},r)
 \big)$, 
 and the threshold
 $
 \delta_M= 
 D(\partial\X,\epsilon/(2L))\big(1 -
 	p_0 \Lambda_{\epsilon}^{r,L}
 \big)^M
 $.
Then, under Assumptions \ref{assum:f:lipschitz},  \ref{assum:Theta:r_convex} and \ref{assum:sampling_density:cor} 
\rev{and assuming that $\partial\Y\subseteq f(\partial\X)$,}  
$$
\Prob(
d_H(
\hat\Y^M, 
\hull(\Y)
)\leq \epsilon
)\geq 1-\delta_M,
\qquad \text{ and }\qquad\ 
\Prob(\Y\subseteq\hat\Y_\epsilon^M)\geq 1-\delta_M.
$$
}
\begin{proof}\textbf{of Corollary \ref{cor:conservative_finite_sample}.}
We prove that Assumptions \ref{assum:Theta:r_convex} and \ref{assum:sampling_density:cor} imply Assumption \ref{assum:sampling_density} with $\Lambda_{\epsilon}^{L}=p_0\Lambda_{\epsilon}^{r,L}$, where $\Lambda_{\epsilon}^{r,L}=\lambda\big(
 B(0,\epsilon/(2L)) 
 \cap B(\vec{r},r)
 \big)$. 
The result then follows by applying Theorem \ref{thm:conservative_finite_sample}.

Specifically, we must prove that 
$\Prob_\X(B(x,\epsilon/(2L)))\geq \Lambda_{\epsilon}^{L}$ for all  
	$x\in\partial\X$. 
	
	From Assumption \ref{assum:sampling_density:cor}, 
	$\Prob_\X(A)\geq p_0\lambda(A)$ for all  
 	$A\in\B(\X)$ for some constant $p_0>0$.

Since $\Prob_\X(\X)=1$ (Assumption \ref{assum:sampling_density:cor}), 
	$\Prob_\X(B(x,\epsilon/(2L)))
	=
	\Prob_\X(\X\cap B(x,\epsilon/(2L)))
	$ for any  
	$x\in\R^n$.

	Therefore, for all  
	$x\in\partial\X$, 
	$\Prob_\X(B(x,\epsilon/(2L)))
	=
	\Prob_\X(\X\cap B(x,\epsilon/(2L)))
	\geq 
	p_0 \lambda(\X\cap B(x,\epsilon/(2L))).$
	
	From Assumption \ref{assum:Theta:r_convex} and Lemma \ref{lem:X_vol_covrg_eps_ball:improved_caps}, we obtain  
	$\Prob_\X(B(x,\epsilon/(2L)))
	\geq 
	p_0 \lambda\big(
B(0,\epsilon/(2L)) 
\cap B(\vec{r},r)
\big)$  for all  
	$x\in\partial\X$ with $\vec{r}=(r,0,\dots,0)\in\R^p$. 
	This concludes the proof of Corollary \ref{cor:conservative_finite_sample}.
\end{proof}

\section{Volume of the intersection of two hyperspheres}
\label{appendix:spherical_caps}\label{apdx:intersection_spherical_cap}
For completeness, we describe the computation of the constant $\Lambda_{\epsilon}^{r,L}=\lambda\big(
B(0,\epsilon/(2L)) 
\cap B(\vec{r},r)
\big)$ in Theorem \ref{thm:conservative_finite_sample} based on the results in \citep{Li2011,Petitjean2013,MattStackExchange2013}. 
Given $r>0$, $a\in\R$, 
and the incomplete beta function 
$I_x(a,b)=\Gamma(a+b)(\Gamma(a)\Gamma(b))^{-1}\int_0^xt^{a-1}(1-t)^{b-1}\dd t$, we define
$$V(r,a)=\begin{cases}
\frac{\pi^{p/2}}{2\Gamma(\frac{p}{2}+1)}r^nI_{1-a^2/r^2}\left(\frac{n+1}{2},\frac{1}{2}\right)\ &\text{if }a\geq 0,
\\
\frac{\pi^{p/2}}{2\Gamma(\frac{p}{2}+1)}r^n(2-I_{1-a^2/r^2}\left(\frac{n+1}{2},\frac{1}{2}\right))
\ &\text{otherwise.}
\end{cases}
$$
Let 
$c_1=\frac{(\epsilon/(2L))^2}{2r}$ and 
$c_1=\frac{2r^2-(\epsilon/(2L))^2}{2r}$. 
Then, 
$
\Lambda_{\epsilon}^{r,L}
=
V(\epsilon/(2L),c_1) 
+
V(r,c_2)$.

\label{appendix:spherical_caps}
\section{Computing the Lipschitz constant of a ReLU network from samples}\label{appendix:lipschitz_relu}
In this section, we show that sampling gradients enables obtaining the Lipschitz constant of a neural network with ReLU activation functions with high probability. 
Consider a feed-forward ReLU neural network $f:\R^n\rightarrow\R^n$ with $\ell\in\mathbb{N}$ layers, given as 
$$
f(x) = W^\ell x^\ell + b^\ell,
\quad x^{k+1}=\phi^k(W^kx^k+b^k), \ k=0,\dots,\ell-1,
\quad 
x^0 = x,
$$
where $(W^k,b^k)_{k=0}^\ell$ are the network weights and biases, and each $\phi^k:\R^{n_k}\rightarrow\R^{n_{k+1}}$ is defined as 
$\phi^k(x)=(\varphi(x_1),\dots,\varphi(x_{n_k}))
$ with $\varphi(z)=\max(0,z)$. 
Note that $f$ is piecewise-affine. 

Let $\X\subset\R^n$ be a non-empty compact set. 
Let $\A=\{A_1,\dots,A_N\}$ be the set of all polytopes $A_i\subseteq\X$ where $f|_{A_i}$ is affine, which we call the activation regions of $f$.  
Let $\Lambda_N=\frac{\min_{i=1,\dots,N} \lambda(A_i)}{\lambda(\X)}$ be the smallest (normalized Lebesgue) volume of all activation regions. 
Note that $f$ is Lipschitz continuous over $\X$, since it is continuous and restricted to a compact subset. Thus, for some $L\geq 0$,
$$
\|f(x_1)-f(x_2)\|\leq L\|x_1-x_2\|\quad \forall x_1,x_2\in\X.
$$
Further, since  $f$ is piecewise-affine, $f$ is $L$-Lipschitz continuous with 
$$
L=\max_{i=1,\dots,N} \{\|\nabla f(x_i)\|\ \text{for some }x_i\in A_i\}.
$$
We propose the following sampling-based method to recover the Lipschitz constant $L$:
\begin{enumerate}\setlength\itemsep{0.5mm}
\item Draw $M$ random samples $x_i$ in $\X$ according to the uniform probability measure over $\X$.
\item Evaluate $L_i=\|\nabla f(x_i)\|$ for all $i=1,\mydots,M$.
\item Set $\hat{L}=\max_{i=1,\dots,M} L_i$.
\end{enumerate}
In general, with this approach, providing statistical guarantees on whether $\hat{L}$ is a valid Lipschitz constant for $f$ is challenging; the analysis would rely on the Hessian of $f$ which is a-priori unknown. 
In this specific setting, $f$ is piecewise-affine, which we leverage in the analysis below.
\begin{lemma}
With the previous notations, define 
$\delta_M=N(1-\Lambda_N)^M
$. 
Then,
$$
\Prob(f\text{ is $\hat{L}$-Lipschitz continuous over $\X$})\geq 1-\delta_M.
$$
\end{lemma}
\begin{proof}
Since $L=\max_{i=1,\dots,N} \{\|\nabla f(x_i)\|\ \text{for some }x_i\in A_i\}$, a sufficient condition for $f$ to be $\hat{L}$-Lipschitz continuous is that at least one point $x_j$ was sampled in each region $A_i$. Thus,
\begin{align*}
\Prob\left(
\bigcap_{i=1}^N\left\{
\bigcup_{j=1}^M \{x_j\}\cap A_i\neq\emptyset
\right\}
\right)
&=
1-
\Prob\left(
\bigcup_{i=1}^N\left\{
\bigcup_{j=1}^M \{x_j\}\cap A_i\neq\emptyset
\right\}^\comp
\right)
\\
&\geq
1-
\sum_{i=1}^N
\Prob\left(
\bigcap_{j=1}^M x_j\notin A_i
\right)
\ \,\quad\text{(Boole's inequality)}
\\
&=
1-
\sum_{i=1}^N
\Prob\left(
 x_j\notin A_i
\right)^M
\qquad\quad\text{(independent samples)}
\\
&\geq 
1-N(1-\Lambda_N)^M,
\end{align*}
where the last step follows from
$\Prob\left(
 x_j\notin A_i
\right)=1-\Prob\left(
 x_j\in A_i
\right)\leq 1-\Lambda_N$. 
\end{proof}

Upper bounds for the number of activation regions $N$, as a function of the number of neurons and layers of the neural network $f$, are available in the literature \citep{Montufar2014,Serra2018,Hanin2019}. 
The regions $A_i$ and their number $N$ could be explicitely computed using formal methods \citep{Serra2018,Vincent2021}. %

\section{Experimental details}
\subsection{Sensitivity analysis}\label{apdx:sensitivity}
We provide more details into the sensitivity analysis in Section \ref{sec:results:sensitivity}.

	\begin{wrapfigure}{R}{0.32\linewidth}
	\begin{minipage}{1\linewidth}
	\vspace{-4mm}
    \centering	\includegraphics[width=1\linewidth]{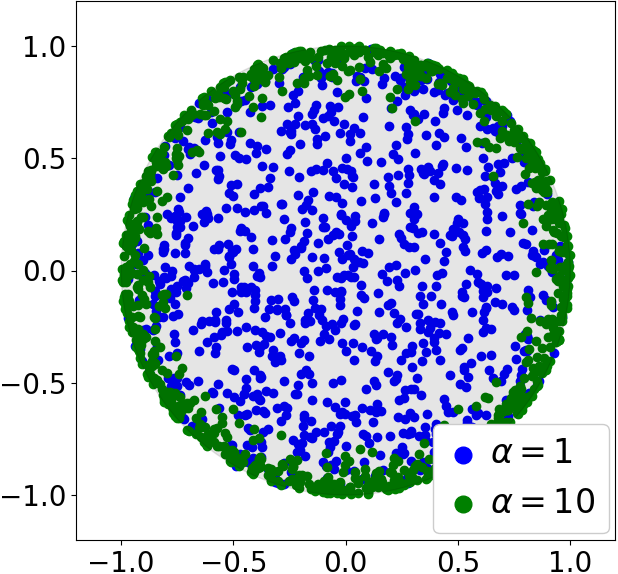}
	\caption{$M=1000$ samples  $x_i$ distributed according to $\Prob_\X^\alpha$ for two values of the parameter $\alpha$.}\label{fig:sensitivity:sampling_distribution}
	\vspace{-16mm}
	\end{minipage}%
	\end{wrapfigure}

\textbf{Sampling distribution}: 
uniformly sampling over a $p$-dimensional ball $\X=B(0,1)$ can be done with the following algorithm (see, e.g., \citep{Harman2010}): %
\begin{enumerate}[leftmargin=5mm]
\item Sample $u\sim\textrm{Unif}(0,1)$ and set the radius $r\gets u^{1/p}$.
\item Sample $z\sim \mathcal{N}(\mathbf{0},\textrm{I}_{p})$ with $\textrm{I}_p\in\R^{p\times p}$ the identity matrix.
\item Set the input sample to $x \gets (r\cdot z)/\|z\|_2$.
\end{enumerate}
Our sampling distribution $\Prob_\X^\alpha$ is a simple modification to this algorithm to yield increasingly larger probabilities of sampling inputs $x$ close to the boundary $\partial\X$. Specifically, we replace the first step above with sampling from a $\beta$-distribution 
$u\sim\textrm{Beta}(\alpha,\beta)$, where $\beta=1$ is fixed and $\alpha\geq 1$ is a varying parameter. Setting $\alpha=1$ yields a uniform distribution (i.e., $\textrm{Beta}(1,1)=\textrm{Unif}(0,1)$) whereas larger values of $\alpha$ yield larger probabilities of sampling close to the boundary $\partial\X$, see Figure \ref{fig:sensitivity:sampling_distribution}.

\textbf{Finite-sample bound}:  	
given a desired $\epsilon$-accuracy, 
evaluating the theoretical coverage probability $1-\delta_M$ from Corollary \ref{cor:conservative_finite_sample} requires evaluating the threshold
 $
 \delta_M= 
 D(\partial\X,\epsilon/(2L))\big(1 -
 	p_0 \Lambda_{\epsilon}^{r,L}
 \big)^M
 $. 
	Since $\X$ is a $2$-dimensional unit-radius ball, the covering term is bounded by $D(\partial\X,\bar\epsilon)\leq (2\pi) / (2\bar\epsilon) + 1$. The volume term $\Lambda_{\epsilon}^{r,L}=\lambda\big(
 B(0,\epsilon/(2L)) 
 \cap B(\vec{r},r)
 \big)$ is computed as described in Section \ref{apdx:intersection_spherical_cap}. 
 The constant $p_0^\alpha>0$ that satisfies
 $\Prob_\X^\alpha(B_x)\geq p_0^\alpha\lambda(B_x)$ for all  
 	$B_x=B(x,\bar\epsilon)\cap\X$ with $x\in\partial\X$ (see Assumption \ref{assum:sampling_density:cor}) can be computed as $p_0^\alpha=(1-\textrm{CDF}_{\textrm{Beta}(\alpha,\beta)}((1-\bar\epsilon)^p)))/(1-\textrm{CDF}_{\textrm{Unif}(0,1)}((1-\bar\epsilon)^p))$. 
Finally, given the fixed threshold $\smash{\delta_M=10^{-3}}$, we compute the corresponding guaranteed accuracy value $\epsilon>0$ that satisfies $
 \delta_M= 
 D(\partial\X,\epsilon/(2L))\smash{(1 -
 	p_0 \Lambda_{\epsilon}^{r,L}
 )^M}
 $ using a bisection method.

\subsection{Verification of neural network controllers}\label{apdx:exps:nn}
We provide further details about the neural network controller experiment in Section \ref{sec:results:verif_closed_nn}. 
In this experiment, we consider the verification of a neural network controller $u_t=\pi_{\textrm{nn}}(x_t)$ for a known linear dynamical system $x_{t+1}=Ax_t+Bu_t$, where $t\in\mathbb{N}$ denotes a time index, and $x_t\in\R^n$ and $u_t\in\R^m$ denote the state and control input. %
Given a set of initial states $\X_0\subset\R^n$, 
the problem consists of estimating the reachable set at time $t\in\mathbb{N}$ defined as
$
\X_t=\{(A(\cdot)+B\pinn(\cdot))\circ\dots\circ 
(Ax_0+B\pi_{\textrm{nn}}(x_0)): x_0\in\X_0\}$. 
Defining $(\X,\Y)=(\X_0,\X_t)$ and $f(x)=(A(\cdot)+B\pinn(\cdot))\circ\dots\circ 
(Ax+B\pi_{\textrm{nn}}(x))$, we see that this problem fits the mathematical form described in Section \ref{sec:intro}. %
In the experiments, we consider $A=\SmallMatrix{
1&1\\0&1
}$, $B=\SmallMatrix{
0.5\\1
}$, %
$
\X_0=\{(x^1,x^2)\in\R^2:
5\leq 2x^1\leq 6, -1\leq 4x^2\leq 1
\}
$, and %
a ReLU network $\pinn$ from \citep{Everett21_journal} with two layers of $5$ neurons each. 
We compare $\epsilon$-\randup with the formal verification technique \reachlp \citep{Everett21_journal} %
and the sampling-based approaches presented in %
\citep{ThorpeL4DC2021} and %
\citep{Gruenbacher2021}. 
We use the Abel kernel $K(x_1,x_2)=\exp(-\|x_1-x_2\|/\scalebox{0.9}{$0.05$})$ for the kernel method \citep{ThorpeL4DC2021} due to its separating property \citep{DeVito2014}. 
To implement \gotube \citep{Gruenbacher2021}, we use the $\epsilon$-\randup algorithm where we replace the last convex hull bounding step with an %
outer-bounding ball. %
We use a uniform sampling distribution for all methods. 
As ground-truth, we use the reachable sets from $\epsilon$-\randup with $\epsilon=0$ and $M=10^6$, which is motivated by the asymptotic results from Theorem \ref{thm:asymptotic_convergence} and was previously done in \citep{Everett21_journal}.

Next, we provide further details into the evaluation of the finite sample bound: given $\epsilon=0.02$, 
sampling $M=1400$ 
inputs that are uniformly-distributed on the boundary $\partial\X$ is sufficient to ensure that the approximated reachable sets from $\epsilon$-\randup are conservative with probability greater than $1-10^{-4}$. 
This result relies on the Lipschitz constant of the closed-loop system, which we set to $L=1$ to evaluate this bound since the neural network controller leads to closed-loop stability. Alternatively, one could use a formal method to compute a bound on this constant (in contrast to using a formal method for reachability analysis, computing this Lipschitz constant only needs to be done once and can be done offline) or sampling-based methods with a large number of samples (see Section \ref{appendix:lipschitz_relu} for an analysis). 
Since the input set is given as $\X=[2.5,3]\times[-0.25,0.25]$, we have $D(\partial\X, \epsilon/(2L))\leq 2\cdot(0.5+0.5)/(2\epsilon/(2L)) + 1=2/\epsilon+1$. 
Finally, since we sample according to a uniform distribution on the boundary and the input set is rectangular, the coverage constant $\Lambda_\epsilon^L$ in Assumption \ref{assum:sampling_density} can be set to $\Lambda_\epsilon^L=(2\cdot(\epsilon/2L)) / (4\cdot0.5)=\epsilon/2$. 
Thus, with $\epsilon=0.02$, 
(which leads to more accurate reachable set approximations than alternative approaches, see Figure \ref{fig:nn_controller:M_vs_dH_time}), from Theorem \ref{thm:conservative_finite_sample}, choosing $M\geq \frac{\log(\delta_M)-\log(D(\partial\X,\epsilon/(2L)))}{\log(1-\Lambda_\epsilon^L)}\approx 1376$ 
is sufficient to be conservative with probability at least $1-10^{-4}$.
\end{document}